\documentclass[10pt]{article}

\usepackage[a4paper, margin=1in]{geometry}

\usepackage{amsthm}
\usepackage{amsmath}
\usepackage{amssymb}
\usepackage{comment}
\usepackage{tikz}
\usepackage{algorithm}
\usepackage{algpseudocode}
\usepackage{outlines}
\usepackage[numbers]{natbib}

\usepackage{enumitem}

\usepackage{arydshln}

\usepackage{pgfplots}
\usepackage{subcaption}
\pgfplotsset{compat=1.18}

\definecolor{acmDarkBlue}{RGB}{0,114,178}
\definecolor{acmGreen}{RGB}{0,158,115}
\definecolor{acmPink}{RGB}{204,121,167}
\definecolor{acmOrange}{RGB}{213,94,0}
\definecolor{acmYellow}{RGB}{240,228,66}
\definecolor{acmLightBlue}{RGB}{86,180,233}

\pgfplotsset{
    cycle list={
        {acmDarkBlue, mark=*},
        {acmGreen, mark=square*},
        {acmPink, mark=triangle*},
        {acmOrange, mark=diamond*},
        {acmYellow, mark=o},
        {acmLightBlue, mark=star}
    }
}

\usepackage{booktabs}
\usepackage{multirow}

\usepackage{hyperref}

\usepackage{graphicx}
\graphicspath{ {./figures/} }

\newcommand*\circled[1]{\tikz[baseline=(char.base)]{
            \node[shape=circle,draw,inner sep=1.5pt] (char) {#1};}}
\newcommand*\dottedcircled[1]{\tikz[baseline=(char.base)]{
            \node[shape=circle,draw,dotted,inner sep=1.5pt] (char) {#1};}}

\usetikzlibrary{arrows,fit,positioning,shapes,chains,fit,shapes,calc,backgrounds}

\usepackage{xcolor, soul}

\usepackage{caption}
\newcommand{\todo}[1]{\textcolor{red}{\textbf{Todo:} #1}}

\theoremstyle{definition}
\newtheorem{definition}{Definition}[section]
\newtheorem{theorem}{Theorem}
\newtheorem{corollary}{Corollary}
\newtheorem{lemma}{Lemma}

\newtheorem{conj}{Conjecture}

\usepackage{authblk}
\usepackage{orcidlink}

\title{Perspectives on Unsolvability in the Roommates Problem\\
{\large Instances are Nearly Solvable and Stable Solutions are Nearly Unique}}
\author[ ]{Frederik Glitzner \orcidlink{0009-0002-2815-6368} and David Manlove \orcidlink{0000-0001-6754-7308}}
\affil[ ]{School of Computing Science, University of Glasgow, Glasgow G12 8QQ, UK}
\affil[ ]{\normalfont \texttt{f.glitzner.1@research.gla.ac.uk, david.manlove@glasgow.ac.uk}}
\date{}

\date{}
\begin{document}

\maketitle

\begin{abstract}
    Instances of the well-studied {\sc Stable Roommates} problem need not admit stable matchings. A long-standing open question posed by Gusfield and Irving (1989) asks about the behaviour of the function $P_n$, which measures the likelihood that a random instance with $n$ agents is solvable (i.e., admits at least one stable matching). While very recently resolved in the limit for the case where $n$ is even and preferences are sampled uniformly at random, this paper provides a comprehensive analysis of the landscape surrounding this question, combining structural, probabilistic, and experimental perspectives. 

    We estimate $P_n$ for instances with preferences sampled from diverse statistical distributions, for even and odd numbers of agents, examining problem sizes up to 5,001 agents, and considering important substructures. Our results reveal that while $P_n$ tends to be low for most distributions, the number and lengths of ``unstable'' structures remain limited, suggesting that random instances are ``close'' to being solvable. Additionally, we present the first empirical study of the number of stable matchings and partitions that random instances admit. Our findings show that the solution sets are typically small, which suggests that many NP-hard problems related to computing optimal stable matchings and partitions become tractable in practice.
\end{abstract}


%\paragraph{Wordcount} 
%9,823 words

%\paragraph{Disclosure Statement}
%The authors report that there are no competing interests to declare. The authors report that generative AI was not used in their research or preparation of this manuscript.

%\clearpage

\section{Introduction}

The {\sc Stable Roommates} problem ({\sc sr}) is a classical combinatorial problem with direct applications in (computational) social choice. Consider a group of friends who want to play one session of tennis, where everyone has preferences over who to play with. Can we match them into pairs such that no two friends prefer to play with each other rather than their assigned partners? The goal here is to find a matching of the agents (or a subset of them) without a \emph{blocking pair} of agents, where each agent in the pair prefers the other over their partner in the matching (or they remain unmatched). If the problem instance $I$, consisting of agents and their preferences, admits such a \emph{stable matching} $M$, then we call $I$ \emph{solvable}. Otherwise, we call $I$ \emph{unsolvable}. Even an instance with as few as three agents may be unsolvable, as \citet{gale_shapley} showed. An instance $I$ of {\sc sr} and a matching $M$ in $I$ can be represented in the form of preference lists and a set of pairs, or as a complete graph and a set of edges. These representations are equivalent, and a formal definition follows.

\begin{definition}[{\sc sr} Instance]
    Let $I=(A,\succ)$ be an {\sc sr} \emph{instance} where $A=\{ a_1, a_2, \dots, a_n \}$, also denoted $A(I)$, is a set of $n\in\mathbb N$ agents and every agent $a_i\in A$ has a strict preference ranking (or \emph{preference relation}) $\succ_i$ over all other agents $a_j\in A\setminus\{a_i\}$. We define the \emph{rank} of an agent $a_i$ according to $a_j$, denoted rank$_j(a_i)$, to be the index of $a_i$ in $a_j$'s preference list (starting at 1). 
\end{definition}

In this paper, we stick to the case where every agent finds every other agent acceptable (also referred to as \emph{complete preferences}), but remark that the case with \emph{incomplete} (truncated) preference lists is also well-studied \cite{matchup}.

\begin{definition}[Stable Matching]
    Let $I=(A,\succ)$ be an {\sc sr} instance. A \emph{matching} $M$ in $I$ is an assignment of some (or all) agents in $A$ into unordered pairs such that each matched agent is contained in exactly one pair. A \emph{blocking pair} of a matching $M$ is a pair of two distinct agents $a_i, a_j\in A$ such that either $a_i$ is unassigned in $M$ or $a_j \succ_i M(a_i)$, and either $a_j$ is unassigned or $a_i \succ_j M(a_j)$, where $M(a_i)$ (respectively $M(a_j)$) is the partner of $a_i$ ($a_j$) in $M$. If $M$ does not admit any blocking pairs, then it is called \emph{stable}.
\end{definition}

The {\sc sr} model draws its relevance from being a very general matching under preferences model, as well as from its many practical applications. As the name suggests, it can model campus housing allocation where two students either share a room or a flat \cite{PPR08}. Furthermore, {\sc sr} can model pairwise kidney exchange markets \cite{roth05}, peer-to-peer networks such as file-sharing networks \cite{GMMR07}, pair formation in chess tournaments \cite{KLM99}, and the well-known bipartite {\sc Stable Marriage} problem \cite{gusfield89} (if allowing incomplete preference lists). The {\sc sr} problem is well-studied \cite{gusfield89, matchup}, and \citet{irving_sr} presented an algorithm that finds a stable matching or determines that none exists in linear time, which is commonly referred to as \emph{Irving's algorithm}. 

A solvable instance is shown in Example \ref{table:solvable}, where $a_1$ ranks the agents $a_2$, $a_5$, $a_3$, etc., in linear order. The problem instance admits the matching $M=\{\{a_1, a_2\}, \{a_3, a_4\}, \{a_5,a_6\} \}$ indicated in circles, the stability of which can be easily verified by hand. However, Example \ref{table:unsolvable} shows an instance with six agents that does not admit any stable matching.

\captionsetup[table]{name=Example}
\begin{table}[!htb]
\centering
\begin{minipage}{.45\linewidth}
\centering
    \begin{tabular}{ c | c c c c c }
    $a_1$ & \circled{$a_2$} & $a_5$ &  $a_3$ & $a_4$ & $a_6$ \\
    $a_2$ & $a_5$ & $a_3$ & \circled{$a_1$} & $a_6$ & $a_4$ \\
    $a_3$ & \circled{$a_4$} & $a_2$ & $a_6$ & $a_1$ & $a_5$ \\
    $a_4$ & $a_1$ & $a_2$ & $a_5$ & \circled{$a_3$} & $a_6$ \\
    $a_5$ & \circled{$a_6$} & $a_2$ & $a_1$ & $a_4$ & $a_3$ \\
    $a_6$ & \circled{$a_5$} & $a_3$ & $a_2$ & $a_4$ & $a_1$
    \end{tabular}
\caption{A solvable {\sc sr} instance}
\label{table:solvable}
\end{minipage}%
\begin{minipage}{.45\linewidth}
\centering
    \begin{tabular}{ c | c c c c c }
    $a_1$ & \circled{$a_2$} & $a_5$ & \dottedcircled{$a_3$} & $a_4$ & $a_6$ \\
    $a_2$ & $a_4$ & \circled{$a_3$} & \dottedcircled{$a_1$} & $a_6$ & $a_5$ \\
    $a_3$ & $a_5$ & $a_4$ & \circled{$a_1$} & \dottedcircled{$a_2$} & $a_6$ \\
    $a_4$ & $a_1$ & \circled{$a_5$} & \dottedcircled{$a_6$} & $a_2$ & $a_3$ \\
    $a_5$ & \circled{$a_6$} & $a_2$ & \dottedcircled{$a_4$} & $a_1$ & $a_3$ \\
    $a_6$ & $a_1$ & $a_2$ & $a_3$ & \circled{$a_4$} & \dottedcircled{$a_5$} 
    \end{tabular}
\caption{An unsolvable {\sc sr} instance}
\label{table:unsolvable}
\end{minipage} 
\end{table}
\captionsetup[table]{name=Table}

In their landmark study of the {\sc Stable Marriage} and {\sc Stable Roommates} problems, \citet{gusfield89} established rich structural properties of these problems and presented many algorithmic results. The authors also posed 12 open questions, most of which have received significant attention since their publication 35 years ago. Problem 8 in their list asks about the limit behaviour of the probability $P_n$ that a random {\sc sr} instance with $n$ agents is solvable -- a question that, despite significant efforts, remained largely unanswered until recently \cite{openproblems19}. The key question is whether $\lim_{n\rightarrow\infty} P_n=0$ or $\lim_{n\rightarrow\infty} P_n>0$. \citet{chin2026randomstableroommatesproblem} recently provided an answer to this question under independent and uniformly distributed preferences and an even number of agents, proving that in the limit, such random instances typically admit no stable matching. In this paper, we take a different route, exploring more fine-grained differences in the typical structures of instances sampled from different types of random distributions of various sizes (rather than in the limit), including instances with odd numbers of agents. These investigations lead to multiple previously undiscovered phenomena.

Another fundamental {\sc sr}-related question posed by \citet{gusfield89} concerns the existence of a succinct certificate for the unsolvability of an instance. The question was answered positively by \citet{tan91_1}, who generalised the notion of a stable matching to a new structure called a \emph{stable partition}. A stable partition is a cyclic permutation $\Pi$ of the agents, where every agent prefers their successor in $\Pi$ over their predecessor in $\Pi$, and no two agents strictly prefer each other over their predecessors (for example, $\Pi=(a_1 \; a_2\; a_3)(a_4\; a_5\; a_6)$ is the unique stable partition of Example \ref{table:unsolvable}, indicated in dotted and unbroken circles). Cycles can be of even or of odd length and, as such, are referred to as \emph{even cycles} and \emph{odd cycles}, respectively.

It is known that random {\sc sr} instances with a large number of agents are unlikely to admit any stable matching due to the likely existence of odd cycles and, even if they are solvable, finding ``optimal'' or ``fair'' stable matchings (in the sense of, for example, finding a stable matching with the maximal number of first choices) is generally NP-hard \cite{CooperPhD}. %Furthermore, computing \emph{almost stable matchings} (referring to matchings with the minimum number of blocking pairs) is a notoriously difficult problem from a computational complexity perspective \cite{abraham06}. 
However, in this paper, we aim to show empirically and, in some cases, structurally, that 
\begin{itemize}
    \item the structures contained in the preferences that make instances unsolvable are few and small, meaning that most instances are very close to admitting a stable matching. This also motivates the use and further study of matchings that are stable within the largest possible group of agents (so-called \emph{maximum stable matchings}, stable half-matchings (in which we allow half-integral assignments), and similar solution concepts for unsolvable instances;
    \item the number of stable matchings and stable partitions, as well as the number of distinct pairs and cycles included within them, is very small for most statistical cultures compared to known adversarial families of instances that admit exponentially many such solutions. This suggests that many NP-hard problems related to computing optimal stable solutions are often quickly solvable in practice.
    %\item although previous work has shown that the minimum number of blocking pairs admitted by a matching in a random instance with preferences chosen uniformly at random is usually very small compared to the instance size, this does not necessarily hold for other statistical cultures. This establishes the conditions under which currently known exponential-time algorithms for the almost stable matching problem are likely to be feasible in practice. On the other side, it also shows where these algorithms are likely to be infeasible, thereby providing further motivation for the study of polynomial-time or faster exponential-time algorithms for this problem.
\end{itemize}

\paragraph{Outline and Contributions}

In this paper, we first review different results surrounding the key solvability question (whether $\lim_{n\rightarrow\infty}P_n=0$ or $\lim_{n\rightarrow\infty}P_n>0$), and the connections between this and related problems studied in the literature, in Section \ref{sec:survey}. Then, in Section \ref{sec:experiments}, we provide novel experimental insights into the structures contained in random {\sc sr} instances with the aim of presenting new perspectives and further intuitions for unsolvable instances. After extending previous experiments that approximate $P_n$ directly for different types of instances with preferences sampled from different statistical distributions, we zoom in on the instances and investigate the likelihood of existence and numbers of different types of stable matchings and stable partitions and the structures contained within them, and then investigate the invariant structures that make instances unsolvable. In Section \ref{sec:implications}, we compare these observations to novel experimental results about matchings that are stable in a large sub-instance (\emph{maximum stable matchings}) and ``optimal stable matchings''. Finally, we conclude the paper with a short summary and an outlook on potential directions of future research regarding {\sc sr} and unsolvable instances more broadly in Section \ref{sec:conclusion}.

Our results suggest that although $P_n$ is expected to be low for most types of instances and sufficiently large $n$, the number of odd cycles is expected to be low, meaning that random instances are nearly stable. While this complies with the recent result that $\lim_{n\rightarrow\infty}P_n=0$, it would imply that the number and sizes of structures causing this instability are small, permitting desirable properties of relaxed solution concepts. Furthermore, although all instances allow some sort of stable structures when relaxing the conditions of a stable matching (for example, to stable partitions), even these are usually nearly unique in the sense that there is very little flexibility and very few such solutions. This renders many NP-hard problems tractable in practice (such as computing stable half-matchings with a maximum number of first choices), motivates easy-to-compute alternative solution concepts similar to stable matchings for unsolvable instances (such as maximum stable matchings), and provides new intuition for further work on solvability questions.

\section{Background and Definitions}
\label{sec:survey}

In this section, we will outline different approaches aimed at understanding random instances and unsolvable structures (Section \ref{sec:solvback}), the definition and results of so-called stable partitions which shed light on the underlying structure of unsolvable instances (Section \ref{sec:sp}), and alternative solution concepts for unsolvable instances and their relation to stable matchings (Section \ref{sec:dealing}).

\subsection{Solvability and Random Instances}
\label{sec:solvback}

In an attempt to resolve the fundamental question posed by \citet{gusfield89} regarding $P_n$, \citet{pittel93problem} proved an asymptotic lower bound on $P_n$ for instances with preferences chosen uniformly at random and showed that the expected number of stable matchings is $\mathbb{E}[S_n]=e^{\frac{1}{2}}$. 
%However, he also noted that his results are not sufficient to settle the key question. In the case that $P_n$ converges, though, its convergence rate is guaranteed to be slow and no faster than $n^{-\frac{1}{2}}$. Furthermore, Pittel established that with super polynomially high conditional probability, the sum of ranks of all agents' partners is close to $n^{\frac{3}{2}}$ for every stable matching and that the highest rank of any agent's partner is of order $n^{\frac{1}{2}}\log n$, implying that a stable matching in {\sc sr} is, in general, likely to be well-balanced (i.e., partners are likely skewed towards the lower ranks). Similarly, 
\citet{pittelirving94} proved an asymptotic upper bound on $P_n$, establishing that $\frac{2e^{\frac{3}{2}}}{\sqrt{\pi n}} \lesssim P_n \lesssim \frac{\sqrt{e}}{2}$. Furthermore, the authors included some empirical evidence on solvability probabilities for instances with preferences chosen uniformly at random, with the number of agents ranging between 100 and 2,000. The results suggest that for this range of $n$, $P_n$ decreases from around 64\% down to around 33\%.
%, indicating that their asymptotic upper bound of around 82\% on $P_n$ is very unlikely to be tight.

On the empirical side, \citet{mertens05} extended smaller-scale experiments performed previously by running extensive Monte Carlo simulations on random {\sc sr} instances sampled from different random graphs (not necessarily assuming complete preferences, unlike in \cite{pittelirving94}) to establish conjectures for complete graphs, grids, and Erdös-Rényi random graphs. On complete graphs with preferences sampled uniformly at random, the results suggest that $P_n \simeq e\sqrt{\frac{2}{\pi}} n^{-\frac{1}{4}}$, where the algebraic decay has strong support from the numerical simulations and the constant is the result of numerical fitting. On grids, the results suggest $P_n=\Theta(q^n)$ where $q<1$ depends on the dimension $d$ of the lattice and the range $r$ of the neighbourhood. Finally, on Erdös-Rényi random graphs with $n$ vertices and edge probability $p$, the results suggest that $P_n \simeq e \sqrt{\frac{2}{\pi}}n^{-\frac{1}{4}}$ which is asymptotically independent of $p$. The authors conclude that the existence of a stable matching depends on the existence of short cycles (in a respective stable partition) of low-degree vertices, which are very likely in grids, leading to the exponential decay in $P_n$.

These results were strengthened ten years later by \citet{mertens15random}. Although the author mainly argued how a simple modification to Irving's original algorithm for {\sc sr} leads to an average-case sub-linear space and time complexity of $O(n^{\frac{3}{2}})$ for an input of size $\Theta(n^2)$, his experiments also involved much larger uniformly random instances than previously studied. The numerical data still supported the previously conjectured algebraic decay of $P_n=\Theta(n^{-\frac{1}{4}})$ on complete graphs, adjusting the constant by roughly 3\% for large instances through a least-squares fit. Overall, the results still suggest that $\lim_{n\rightarrow\infty} P_n=0$. 

In a different paper, \citet{mertens15small} built on the ideas presented by \citet{pittel93problem} and presented new formulae for the exact computation of $P_n$. He also presented exact results for $P_n$ with the number of agents $n<13$ and concluded that $P_n$ is much smaller for $n$ odd, where an instance with an odd number of agents is considered stable if it contains one fixed point, i.e., one unmatched agent, and the remaining agents form a stable matching. We will also observe this in our experiments later.
%However, the author also noted that although his computation is much more efficient than exhaustive enumeration, it is not a feasible method for computing much larger values of $P_n$ exactly to rigorously determine its limiting behaviour.

In a recent report, \citet{openproblems19} briefly summarised the advances on the solvability problem and confirmed that there is not enough insight yet to determine the ultimate behaviour of $P_n$ as $n$ grows large and that new insights are likely infeasible without a new approach. However, as already noted before, \citet{chin2026randomstableroommatesproblem} proved very recently in a preprint that, when the number of agents is even, and preferences are complete, independent, and drawn uniformly at random, $P_n\rightarrow n^{-1/17}$, thus proving that $P_n$ converges to zero in this specific setting.

The study of random instances has focused almost exclusively on instances with preferences sampled uniformly at random. However, \citet{boehmer2024map,diversesynthetic23} studied the characteristics of and relationships between randomly generated instances sampled from different statistical cultures. The authors introduced a mutual attraction distance to measure the similarities between random {\sc sr} instances and four extreme cases of preference lists (according to their distance measure). They then defined ten different statistical cultures, one of them being preferences sampled uniformly at random, and plotted 460 instances with 200 agents each based on their mutual attraction distance against the four extrema. The plots clearly show that instances from the same culture have very similar properties; however, these properties can vary greatly from instances sampled from other cultures, especially concerning solvability. Our experiments greatly extend these simulations and present a systematic study of instances with varying sizes, statistical cultures, and deeper insights into the structures present in those instances.

\citet{yilmaz2025} also provided a novel framework for generating random {\sc sr} instances, with a focus on generating instances with a large number of stable matchings. However, their generation only ever generates solvable instances, while we focus explicitly on instances that may or may not be solvable.

We note that an entirely different perspective on {\sc sr} experiments was provided very recently by \citet{Molis2026}, who conducted a lab-based experiment with real human participants, who entered into a decentralised real-time partner search. While their results have no direct connection to ours, they also discuss alternative solution concepts for unsolvable instances, albeit from a decentralised angle.

\subsection{Stable Partitions}
\label{sec:sp}

As previously noted, \citet{gusfield89} asked whether it is possible to provide a succinct certificate for the unsolvability of an {\sc sr} instance. The question was answered positively by \citet{tan91_1}, who introduced the combinatorial structure known as a \emph{stable partition}, although the structure's significance extends beyond its use as a witness of unsolvability.

\begin{definition}[Stable Partition]
    Let $I=(A,\succ)$ be an {\sc sr} instance. Then a partition $\Pi$ is \emph{stable} if it is a permutation of $A$ and  
    \begin{enumerate}
        \item[(T1)] $\forall a_i \in A$ we have $\Pi(a_i) \succeq_i \Pi^{-1}(a_i)$, and
        \item[(T2)] $\nexists.a_i, a_j \in A, \; a_i\neq a_j,$ such that $a_j \succ_i \Pi^{-1}(a_i)$ and $a_i \succ_j \Pi^{-1}(a_j)$,
    \end{enumerate}
    where $\Pi(a_i) \succeq_i \Pi^{-1}(a_i)$ means that either $a_i$'s successor in $\Pi$ is equal to its predecessor, or the successor has a better rank than the predecessor in the preference list of $a_i$.
\end{definition}

Over 20 years after their initial publication, \citet{matchup} referred to the work on stable partitions in the 1990s as a key landmark in the progress made on the {\sc sr} problem after 1989. 

Notice that the definition of a stable partition is a relaxation of that of stable matchings: if we changed T1 to ``$\forall a_i \in A$ we have $\Pi(a_i) = \Pi^{-1}(a_i)$'' then we would recover permutations equivalent to stable matchings (which might not exist and could include a fixed point when $n$ is odd). Being permutations, stable partitions can be written as an unordered collection of disjoint ordered cycles. For clarity, we distinguish between different cycles as follows.

\begin{definition}
    A cycle $C$ of length $k$ is referred to as a \emph{$k$-cycle}. If $k$ is odd, then we refer to $C$ as an \emph{odd cycle}; otherwise, we refer to $C$ as an \emph{even cycle}. Furthermore, $C$ is called \emph{reduced} if $k$ is either 2 or odd; otherwise, it is called \emph{non-reduced}. Similarly, a stable partition $\Pi$ is \emph{reduced} if it consists only of reduced cycles, and \emph{non-reduced} if not.
\end{definition}

\citet{tan91_1,tan91_2} showed that any {\sc sr} instance admits at least one reduced stable partition (in stark contrast to stable matchings) and established the following properties.

\begin{theorem}[\cite{tan91_1, tan91_2}]
\label{thm:tan91}
    The following properties hold for any {\sc sr} instance $I$.
    \begin{itemize}
        \item Any two stable partitions $\Pi_a, \Pi_b$ of $I$ contain exactly the same cycles of odd length.
        \item $I$ admits a stable matching if and only if no stable partition of $I$ contains a cycle of odd length at least 3.
        \item Given a stable partition $\Pi$, any cycle of even length longer than 2 can be broken down into a collection of cycles of length 2 to achieve a reduced stable partition $\Pi'$.
    \end{itemize}
\end{theorem}

In his original paper, \citet{tan91_1} provided a linear-time algorithm similar to Irving's algorithm to compute a stable partition, referred to as \emph{Tan's algorithm}. \citet{tanhsueh} considered the online version of the problem of finding a stable partition, in which a new agent arrives and the preference lists are updated, and constructed an exact algorithm known as the \emph{Tan-Hsueh algorithm} that runs in linear time (for each newly arriving agent).

%\todo{integrate new result}
Returning to solvability and the key question concerning $P_n$, \citet{pittel93instance} derived a range of relevant probabilistic results relating to the original algorithm by Tan \cite{tan91_1} when assuming preferences sampled uniformly at random. Specifically, he showed, for example, that every stable partition is likely to be almost a stable matching, in the sense that at most $O(\sqrt{(n\log n)})$ members are likely to be involved in odd cycles of length 3 or more. He also showed that the expected number of stable partitions is $O(\sqrt{n})$. 

Experimentally, \citet{mertens05} analyzed the total number of elements in cycles of odd length, denoted by $n_{odd}$, of unsolvable instances sampled from complete uniform random graphs and conjectured that the expected total number is $\Theta\left (\sqrt{\frac{n}{\log n}}\right )$, with the numerical constant estimated to be around 2.375. Furthermore, on the experimental side, \citet{mertens15small} presented the exact probabilities for specific cycle types for fixed instance sizes through exhaustive enumeration. For example, the probability that an instance with 10 agents admits a stable partition consisting of five 2-cycles, denoted by $P([2^5])$, is 0.0013. Some other combinations are also explored, such as the probability of a stable partition consisting of one 1-cycle, three 2-cycles, and one 3-cycle $P([1^1, 2^3, 3^1])=0.0000$. However, Mertens did not provide any results on the probabilities of fixed odd-length cycles in isolation; for example, what is the probability of an instance of size $n$ admitting a stable partition with an invariant cycle of (odd) length $x$?

Some intuition for this question is given in recent work by \citet{pittel19}, extending previous work \cite{pittel93instance} in a similar, deeply probabilistic and algorithmic investigation. Here, he studied stable partitions where preferences are chosen uniformly at random and proved that the expected total number of reduced stable partitions grows in the order of $n^{\frac{1}{4}}$, which extends the previous $O(\sqrt{n})$ bound on the expected number of stable partitions. Pittel also showed that the expected total number of odd-length cycles grows in the order of at most $n^{\frac{1}{4}}\log n$. This is an interesting contrast to the previously proven result that the expected number of stable matchings approaches $e^{\frac{1}{2}}$. Furthermore, Pittel showed that with super-polynomially high probability, $n_{odd}$ (the number of agents contained in cycles of odd length) is below $\sqrt{n}\log n$ and that the probability of a stable partition having a \emph{fixed-point} (1-cycle) is bounded by $O(n^2e^{-\sqrt{n}})$. However, all of these results are derived for instances with instances chosen uniformly at random and $n$ even, so a natural question is whether these results also hold when sampling the preferences from other statistical distributions or when $n$ is odd, which we answer experimentally in this paper.

Recently, \citet{glitzner2024structuralalgorithmicresultsstable} established new algorithmic and structural results for stable partitions and showed, for example, how to efficiently enumerate all stable partitions and various cycle types found within them. The authors established that although an {\sc sr} instance $I$ with $n$ agents can admit exponentially many stable matchings (and thus stable partitions), it admits at most $O(n^2)$ stable cycles, and these can be enumerated in $O(n^4)$ time. Furthermore, the authors showed that reduced stable partitions stand in a bijective correspondence with the stable matchings of a smaller solvable sub-instance $I_T$ of $I$ such that the reduced stable partitions of $I$ can be enumerated at least as efficiently as the stable matchings of $I_T$, while all stable partitions of $I$ can be enumerated with an asymptotic factor of $n^2$ slower. In the paper, the authors also adapted optimality criteria from stable matchings to stable partitions and gave complexity and approximability results for the problems of computing such ``fair'' and ``optimal'' stable partitions, establishing that several of these problems are NP-hard, as is the case in the corresponding problems in the stable matching setting. We will use some of their algorithms in this work and greatly extend their preliminary experimental results on stable partitions in terms of the questions being asked, the different statistical distributions to sample preferences from, the ranges of agent numbers in the instances, and provide new formal results and conjectures.

\citet{glitznergspteac} also recently extended the stable partition structure to more general settings where agents have capacities and strengthened known properties from the one-to-one setting to such settings.

\subsection{Dealing with Unsolvable Instances}
\label{sec:dealing}

From the previous subsection, it is clear that as $n$ grows large, stable matchings are unlikely to be a consistent solution concept in practice, as they are unlikely to exist. In the past, many alternative problems and solutions have been proposed, such as stable partitions \cite{tan91_1}, maximum stable matchings \cite{tan91_2}, almost stable matchings \cite{abraham06}, and many more. We will highlight some of them, referring to \citet{matchup} for more details and examples.

\citet{tan91_2} introduced the notion of a \emph{maximum stable matching}, a matching of maximum size such that no pair of agents both having a partner in the matching are blocking. Given a stable partition, finding such a maximum stable matching becomes simple and solvable in $O(n^2)$ time (for an instance with $n$ agents). To compute a maximum stable matching, one can simply pick an arbitrary agent from each odd cycle, delete it from the instance, and decompose the remaining even-length cycle into transpositions. However, this might not be a good solution in practice, as up to a third of the agents could remain unmatched, even when all preference lists are complete. Interestingly, \citet{manipulation24} recently proved that when only agents from a specific set can be removed (rather than any agent from the instance), computing a maximum stable matching is NP-hard.

It has also been shown that each stable partition corresponds to a \emph{stable half-matching}, i.e., a half-integral fractional matching. Again, a given stable partition can be easily converted into a stable half-matching -- we can simply match every agent one half-unit to their predecessor and one half-unit to their successor (resulting in a full match if and only if the agent is in a transposition). Stable half-matchings have a variety of practical motivations, for example in sports scheduling, and have been studied in various matching models \cite{biro08,birofleiner15}.

\section{Analysing Unsolvable Structures}
\label{sec:experiments}

Most existing results outlined in the previous section are either theoretical proofs of the asymptotic runtime and behaviour of classical algorithms, or experimental estimates of the solvability probability itself. Inspired by the work of \citet{diversesynthetic23} and motivated by the recent study of instance and stable partition structures \cite{glitzner2024structuralalgorithmicresultsstable}, we aim to go deeper and investigate various properties of stable partitions and structures contained within them for different statistical cultures. 
%With this, we hope to accelerate progress on open questions such as the approximability of the almost stable matching problem and the behaviour of $P_n$ in the limit \cite{pittelirving94}. \todo{integrate new result}

First, we will outline our experimental design (Section \ref{sec:design}), while the remainder of this section will mirror a ``zooming-in'' process: we will start by investigating how many instances we expect to admit stable matchings (Section \ref{sec:solv}), and then analyse how many stable partitions and reduced stable partitions instances admit (Section \ref{sec:numpartitions}). We will then see how much these structures vary between each other (Section \ref{sec:variability}) and how this compares to the case when we restrict attention to stable matchings (Section \ref{sec:nummatchings}). Finally, we will further investigate the types of cycles that likely make up our stable partitions (Section \ref{sec:oddcycles}). %and how far the unsolvable instances are from being solvable (Section \ref{sec:bp}).

\subsection{Experimental Design}
\label{sec:design}

For our experiments, we generated random {\sc sr} instances based on different statistical cultures previously proposed by \citet{boehmer2024map,diversesynthetic23}. We chose a wide-ranging subset of these cultures, with informal descriptions following, referring to the referenced paper for rigorous definitions:
\begin{itemize}
    \item \textbf{\textit{IC:}} Impartial culture, used synonymously with preferences generated uniformly at random.
    \item \textbf{\textit{2-IC:}} The set of agents is partitioned into two equally sized groups. Every agent prefers all agents from its group over all agents from the other group, but the preferences over the agents within the groups are generated uniformly at random. This could model left- versus right-wing political leanings, for example.
    \item \textbf{\textit{Symmetric:}} Preferences agree, so  rank$_{j}(a_i)=$ rank$_{i}(a_j)$. This could model the phenomenon in which similarities attract.
    \item \textbf{\textit{Asymmetric:}} Preferences disagree, so rank$_{j}(a_i)=k$ implies rank$_{i}(a_j)=n-k$. This could model the phenomenon in which opposites attract.
    \item \textbf{\textit{Euclidean:}} Agents are randomly distributed in a (two-dimensional) Euclidean space and strict preferences are derived from their Euclidean distances (with arbitrary tie-breaking). This could model the phenomenon in which preferences are based on spatial distance.
    \item \textbf{\textit{Attributes:}} Weighted (two-dimensional) \textit{Euclidean} preferences where each agent has a random position and random weighting over the spatial dimensions. This could model the case where each dimension models some personality trait and the weight models personality-based preferences over the traits. 
    \item \textbf{\textit{Mallows-Euclidean:}} Randomly perturbed (two-dimensional) \textit{Euclidean} preferences\footnote{We used the default parameters \texttt{space='uniform', phi=0.5}.} (later also referred to as \textit{M-Euclidean}). We do not have an intuitive use case for this, but \citet{diversesynthetic23} showed that this culture can act as an interesting extreme against the others mentioned above.
\end{itemize}

The controls in most of our experiments are whether there is an even or odd number of agents, whether the sample space is restricted to solvable, unsolvable, or all instances, and which statistical culture we sample the instances from. We generated instances of size (referring to the number of agents $n$) 2-201 (in steps of 1) and instances of size 300, 301, 400, 401, 500, 501, 5,000, and 5,001 agents. For each instance size, we generated 7,000 (seeded) random instances for each statistical culture over which the following results are averaged. For each instance, we computed a stable partition using a custom implementation of the Tan-Hsueh algorithm \cite{tanhsueh} and captured the numbers and lengths of cycles observed. Furthermore, for some instance sizes, we enumerated their reduced and non-reduced stable cycles and partitions using recently developed algorithms by \citet{glitzner2024structuralalgorithmicresultsstable}. % and computed almost stable matchings using a modified version of the {\tt kBP} algorithm by \citet{abraham06}. 
All implementations were written in Python, and all computations were performed on the {\tt fatanode} cluster.\footnote{See \href{https://ciaranm.github.io/fatanodes.html}{https://ciaranm.github.io/fatanodes.html} for technical specifications.}

\subsection{Solvability}
\label{sec:solv}

As mentioned in Section \ref{sec:solvback}, \citet{mertens15small} already computed the exact probability $P_n$ that a random instance with $n$ agents is solvable for small $n$, both even and odd, using exhaustive instance enumeration. The results are shown in Table \ref{table:mertensexact} and indicate that the decay of $P_n$ is much steeper for $n$ odd compared to $n$ even. The author suggests that this is likely because a fixed-point is an agent that would be happy to be matched to any other agent in the instance, which can be highly destabilising.

\begin{table}[!htb]
    \centering
    \begin{tabular}{c c c c c c c c c}
        \toprule
        & \multicolumn{4}{c}{$n$ even} & \multicolumn{4}{c}{$n$ odd} \\
        \cmidrule(lr){2-5} \cmidrule(lr){6-9}
        & $n=2$ & $n=4$ & $n=6$ & $n=10$ & $n=3$ & $n=5$ & $n=7$ & $n=11$ \\
        \midrule
        $P_n$ (in \%) & 100 & 96.30 & 93.33 & 89.13 & 75.00 & 58.96 & 47.54 & 32.39 \\
        \bottomrule
    \end{tabular}
    \caption{Exact values of $P_n$ computed by \citet{mertens15small}}
    \label{table:mertensexact}
\end{table}

Our numerical estimates of $P_n$ (i.e., the proportion of solvable instances in our dataset, denoted by $\hat{P}_n$) are shown in Figure \ref{fig:Pn_subplots} for different statistical cultures. Our results for \emph{IC} (shown in the top left of Figure \ref{fig:Pn_subplots}) match those known (for example by \citet{mertens05, mertens15random}) closely. Furthermore, \citet{mertens05} computed an estimate of $P_n$ for $n$ even through a function of best fit of the form $\hat{P_n} \simeq an^b$ with $a,b\in \mathbb{R}$ and found that $P_n \simeq e\sqrt{\frac{2}{\pi}} n^{-\frac{1}{4}}$. Using the same technique, we complement this result by fitting a function of best fit to our data for $n$ odd, yielding the following conjecture.

\begin{conj}
    If $n$ is odd, and instances are sampled uniformly at random, then $P_n \simeq e\sqrt{\frac{3}{\pi}}n^{-1}$.
\end{conj}

Although previous experiments have largely focused on instances with preferences chosen uniformly at random (the \emph{IC} culture) \cite{pittelirving94,mertens15random}, \citet{diversesynthetic23} showed that there is a clear discrepancy between the results when choosing different statistical cultures. 

Our experiments suggested that instances with an even number of agents and preferences sampled from the \textit{Symmetric}, \textit{Asymmetric}, and \textit{Euclidean} cultures always admit a stable matching. \textit{Euclidean} instances also always admit a stable matching when the number of agents is odd, whereas \textit{Asymmetric} instances never do. This makes sense -- in the \textit{Symmetric} case, every agent can be matched to their first choice (recall that if $a_i$ is $a_j$'s first choice, then $a_j$ must be $a_i$'s first choice). For \emph{Asymmetric}, we show that the following holds.

\begin{lemma}
\label{lemma:asymmetric}
    Let $I$ be an {\sc sr} instance with $n$ agents and asymmetric preferences (sampled from the \emph{Asymmetric} culture). If $n$ is even, then $I$ is solvable. If $n$ is odd, then $I$ is not solvable.
\end{lemma}
\begin{proof}
    Recall that preference lists are of length $n-1$ (by complete preference lists).
    
    If $n$ is even, then each preference list is of odd length. Furthermore, for every agent $a_i$, there exists an agent $a_j$ in the middle of the preference list of $a_i$ (we will refer to them as the \emph{middle agent} of $a_i$, denoted by middle$(a_i)$) such that, by our asymmetric assumption, rank$_i(a_j)=$ rank$_j(a_i)$, and for all agents $a_k$ with rank$_i(a_k)<$ rank$_i(a_j)$ (i.e., $a_i$ prefers $a_k$ to $a_j$), we have that rank$_k($middle$(a_k))<$ rank$_k(a_i)$ ($a_k$ prefers their middle agent to $a_i$). Thus, if $n$ is even, every agent can be matched to their middle agent, and no two agents can be blocking (otherwise, these two agents would prefer each other to their respective middle agents -- a contradiction). 

    If $n$ is odd, then there is no such middle agent (because $n-1$ is even). However, in this case, $I$ admits a unique stable partition $\Pi$ of $I$ containing just one odd cycle of length $n$. Specifically, consider $\Pi=(a_{i_1} \; a_{i_2} \dots a_{i_n})$ such that for each agent $a_{i_j}$, $a_{i_{j+1}}$ is chosen such that rank$_{i_j}(a_{i_{j+1}}) = \frac{n-1}{2}$ (all agent indices modulo $n+1$). Then, by asymmetric preferences, rank$_{i_{j+1}}(a_{i_{j}}) = \frac{n+1}{2}$ and $\Pi$ satisfies stability condition T1. To see that $\Pi$ also satisfies T2 (and is indeed a stable partition), consider two distinct agents $a_j,a_k$ that strictly prefer each other to their predecessors in $\Pi$. Then, by construction, rank$_j(a_k)<$ rank$_j(\Pi^{-1}(a_j))=\frac{n+1}{2}$ and rank$_k(a_j)<$ rank$_k(\Pi^{-1}(a_k))=\frac{n+1}{2}$, i.e., rank$_j(a_k)\leq \frac{n-1}{2}$ and rank$_k(a_j)\leq \frac{n-1}{2}$. However, then rank$_j(a_k)+$rank$_k(a_j)\leq n-1$, but by definition of asymmetric preferences, rank$_j(a_k)=n-$rank$_k(a_j)$, i.e., rank$_j(a_k)+$rank$_k(a_j)=n$, a contradiction. Thus, $\Pi$ is stable, and because it only consists of an odd-length cycle, by Theorem \ref{thm:tan91}, it is the unique stable partition of $I$. Hence, $I$ is unsolvable.   
\end{proof}

The \textit{Euclidean} case is similar, where agents can be matched increasingly by distance, therefore never producing a blocking pair (and leaving one agent unmatched at the end if $n$ is odd). This has been shown by \citet{arkin09}. Note that for an odd number of agents, the \textit{\textit{Symmetric}} preference construction does not make sense (even for $n=3$, two agents must be each other's first choice, but then both of these agents must rank the third agent in position 2 of their preference list, thereby forcing the third agent to rank both of them in position 1, a contradiction).

More interesting are our results shown in Figure \ref{fig:Pn_subplots}. \textit{2-IC} shows a phenomenon where $\hat{P_n}$ seems to jump alternatingly between even values of $n$, where the upper peaks roughly match the behaviour from the impartial culture. We suspect that this is because the two sets that are impartially selected from (in order) are of size $\frac{n}{2}$, so when their size is odd, a cycle of odd length is more likely to occur, whereas potential cycles of even length caused by the partition can be broken up into transpositions. Finally, $\hat
{P}_n$ also decays for the \textit{Mallows-Euclidean} and \textit{Attributes} cultures, with the former decaying at a rate slightly faster than the latter. An interesting distinction between these two cultures compared to \emph{IC} and \emph{2-IC} is that there is no significant difference between our estimates for $n$ even and odd (for sufficiently large $n$).

\begin{figure}[!hbt]
    \centering

    % Subplot 1
    \begin{subfigure}{0.49\textwidth}
        \begin{tikzpicture}
            \begin{axis}[
                width=\textwidth,
                height=5cm,
                ylabel={$\hat{P}_n$},
                ymin=0,       % Set the minimum y-axis value
                ymax=1,         % Set the maximum y-axis value
                xmin=0,
                xmax=201,
                grid=both,
                grid style={dashed, gray!30},
                cycle list name=color list,
                every axis plot/.append style={thick},
                title={\textit{IC}},
                title style={
                    yshift=-1.5ex  % Adjust to bring title closer to the plot
                },
                legend cell align={left},
                legend style={
                    at={(1.28,1.35)}, % Adjusts the legend position
                    anchor=north east,
                    column sep=1ex, % Space between legend columns
                },
            ]
            \addplot[acmDarkBlue, mark=*] table [x=n, y=Pn] {data/solvability/icEven.txt};
            \addplot[acmGreen, mark=square*] table [x=n, y=Pn] {data/solvability/icOdd.txt};
            \addlegendentry{$n$ even}
            \addlegendentry{$n$ odd}
            \end{axis}
        \end{tikzpicture}
    \end{subfigure}
    \hfill
    % Subplot 2
    \begin{subfigure}{0.49\textwidth}
        \begin{tikzpicture}
            \begin{axis}[
                width=\textwidth,
                height=5cm,
                ymin=0,       % Set the minimum y-axis value
                ymax=1,         % Set the maximum y-axis value
                xmin=0,
                xmax=201,
                grid=both,
                grid style={dashed, gray!30},
                cycle list name=color list,
                every axis plot/.append style={thick},
                title={\textit{2-IC}},
                title style={
                    yshift=-1.5ex  % Adjust to bring title closer to the plot
                },
            ]
            \addplot[acmDarkBlue, mark=*] table [x=n, y=Pn] {data/solvability/2icEven.txt};
            \addplot[acmGreen, mark=square*] table [x=n, y=Pn] {data/solvability/2icOdd.txt};
            \end{axis}
        \end{tikzpicture}
    \end{subfigure}

    \vspace{0.1cm}

    % Subplot 3
    \begin{subfigure}{0.49\textwidth}
        \begin{tikzpicture}
            \begin{axis}[
                width=\textwidth,
                height=5.5cm,
                xlabel={Number of Agents ($n$)},
                ylabel={$\hat{P}_n$},
                ymin=0,       % Set the minimum y-axis value
                ymax=1,         % Set the maximum y-axis value
                xmin=0,
                xmax=201,
                grid=both,
                grid style={dashed, gray!30},
                cycle list name=color list,
                every axis plot/.append style={thick},
                title={\textit{Attributes}},
                title style={
                    yshift=-1.5ex  % Adjust to bring title closer to the plot
                },
            ]
            \addplot[acmDarkBlue, mark=*] table [x=n, y=Pn] {data/solvability/attributesEven.txt};
            \addplot[acmGreen, mark=square*] table [x=n, y=Pn] {data/solvability/attributesOdd.txt};
            \end{axis}
        \end{tikzpicture}
    \end{subfigure}
    \hfill
    % Subplot 4
    \begin{subfigure}{0.49\textwidth}
        \begin{tikzpicture}
            \begin{axis}[
                width=\textwidth,
                height=5.5cm,
                xlabel={Number of Agents ($n$)},
                ymin=0,       % Set the minimum y-axis value
                ymax=1,         % Set the maximum y-axis value
                xmin=0,
                xmax=201,
                grid=both,
                grid style={dashed, gray!30},
                cycle list name=color list,
                every axis plot/.append style={thick},
                title={\textit{M-Euclidean}},
                title style={
                    yshift=-1.5ex  % Adjust to bring title closer to the plot
                },
            ]
            \addplot[acmDarkBlue, mark=*] table [x=n, y=Pn] {data/solvability/MEuclEven.txt};
            \addplot[acmGreen, mark=square*] table [x=n, y=Pn] {data/solvability/MEuclOdd.txt};
            \end{axis}
        \end{tikzpicture}
    \end{subfigure}

    \caption{Estimates of $P_n$ by statistical culture}
    \label{fig:Pn_subplots}
\end{figure}
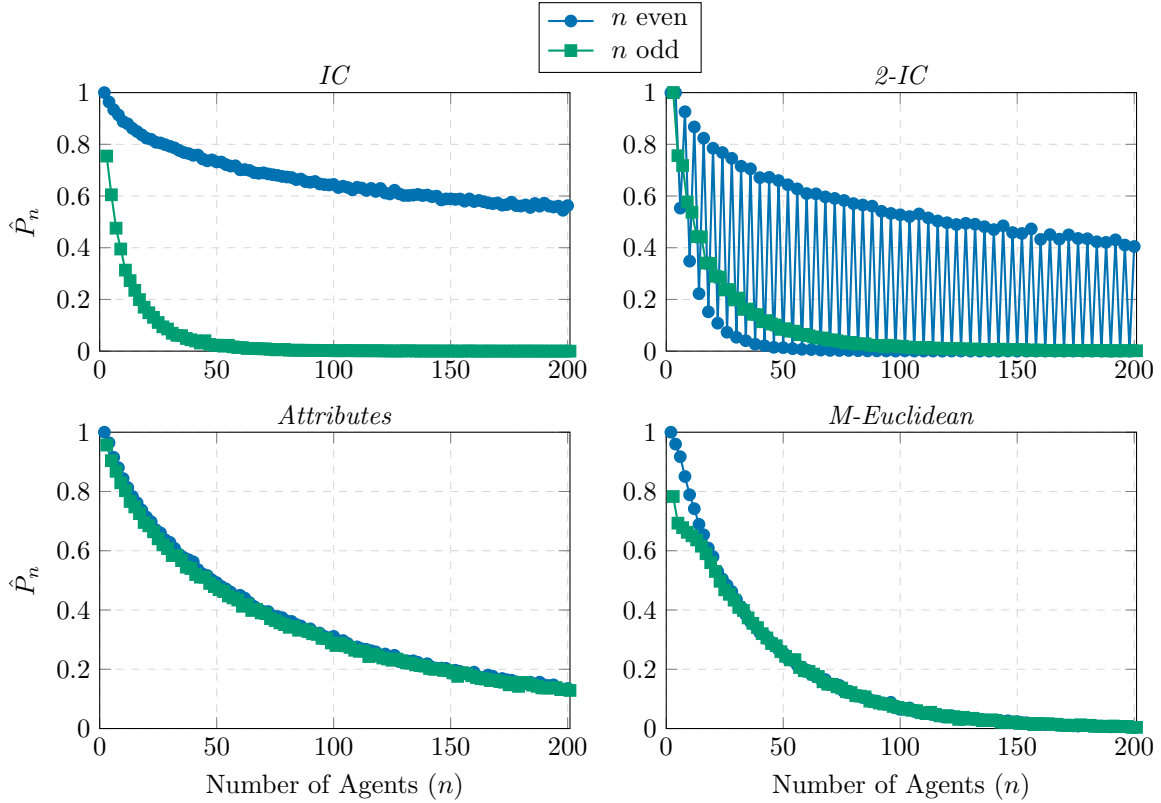

Overall, these results extend some preliminary observations by \citet{boehmer2024map,diversesynthetic23}, but provide much more detailed observations about the solvability of instances sampled from these cultures in different settings and more rigorous estimates for the respective behaviour of $P_n$. In particular, the results appear to be in line with the asymptotic limit $\lim_{n\rightarrow\infty}P_n=0$ \cite{chin2026randomstableroommatesproblem}, not just for uniformly random preferences (\emph{IC}) with even numbers of agents, but even for the different and more structured types of distributions that we considered.%\todo{integrate result}

\subsection{Number of Stable Partitions}
\label{sec:numpartitions}

Now that we have seen that there is a wide range of estimates for $P_n$ depending on the underlying statistical culture, we relax our experiments from investigating only the existence of stable matchings to studying different types of stable partitions (rather than stable matchings) that the instances admit instead. Recall that if an instance is solvable (i.e., admits at least one stable matching), then its stable matchings (denoted by $\mathcal M$) are equivalent to its reduced stable partitions (denoted by $\mathcal{RP}$) (in this case containing only cycles of length 2 and, for $n$ odd, one fixed point). \citet{glitzner2024structuralalgorithmicresultsstable} extended this result and showed that even for unsolvable instances, the reduced stable partitions stand in a bijective correspondence with the set of stable matchings of an underlying solvable sub-instance. Note that any instance admits at least one reduced stable partition (this follows from Theorem \ref{thm:tan91}) and that the set of all stable partitions (denoted by $\mathcal P$) contains all reduced stable partitions (and potentially some non-reduced stable partitions). Given that $\mathcal M\subseteq \mathcal{RP}\subseteq \mathcal P$, it follows that $\vert \mathcal M\vert\leq \vert \mathcal{RP}\vert \leq \vert \mathcal P\vert$. Also, note that it is easy to show that even $\mathcal M$ can be exponentially large (in the number of agents) \cite{gusfield89}. 

It is natural to ask how many reduced and non-reduced stable partitions our random instances admit, and how this varies between statistical cultures, as this impacts many search and optimisation problems related to stable solutions. Tables \ref{table:numstructureseven}-\ref{table:numstructuresodd} show the average number of stable partitions and reduced stable partitions for different cultures and for even and odd numbers of agents.\footnote{The tables and the following results exclude 18 (out of $>$196,000) \emph{Attributes} and 43 (out of $>$196,000) \emph{Mallows-Euclidean} instances where the enumeration timed out after 20h of computation time.}

The most obvious initial observation is that $\vert \mathcal P\vert = 1$ for \emph{Symmetric, Asymmetric}, and \emph{Euclidean} (no matter whether $n$ is even or odd), so we always observed a unique stable partition which is also reduced. \citet{arkin09} already proved theoretically that \emph{Euclidean} instances admit unique stable matchings, and our observations for \emph{Symmetric} and \emph{Asymmetric} instances are consistent with what we previously observed and argued in Section \ref{sec:solv}. 

\begin{lemma}
    Let $I$ be an {\sc sr} instance with an even number of agents $n$ and \emph{Symmetric} preferences, then $I$ admits a unique stable partition.
\end{lemma}
\begin{proof}
    Consider the matching $M$ in which every agent is paired with their first choice. $M$ is a perfect matching by assumption of symmetric preferences. Furthermore, it is stable because no agent can strictly prefer any other agent to their first preference list entry (which they are matched to). Indeed, for any other matching $M'\neq M$,  there would be at least two agents that are each other's first choices but are matched to someone worse (or remain unmatched). Hence, $M'$ could not be stable, and $M$ is unique. By correspondence with stable partitions (Theorem \ref{thm:tan91}), $M$ corresponds to a reduced stable partition and is unique.
\end{proof}

An interesting observation about \emph{Asymmetric} instances is the following.

\begin{lemma}
    Let $I$ be an {\sc sr} instance with $n$ agents and \emph{Asymmetric} preferences. If $n$ is odd, then $I$ admits a unique stable partition, however, this does not necessarily hold for $n$ even.
\end{lemma}
\begin{proof}
    We established in Lemma \ref{lemma:asymmetric} that if $n$ is odd, then there exists a stable partition $\Pi$ consisting of an $n$-cycle, so by Theorem \ref{thm:tan91}, it is the unique stable partition.

    For $n$ even, we showed that there is a stable partition $\Pi$ consisting only of transpositions of the middle choices. However, to see that this is not always unique, consider the instance shown in Example \ref{table:asymmetric} with $n=6$ agents and \emph{Asymmetric} preferences.

    \captionsetup[table]{name=Example}
    \begin{table}[!htb]
        \centering
        \begin{tabular}{ c | c c c c c}
            $a_1$ & $a_6$ & $a_3$ & $a_2$ & $a_4$ & $a_5$ \\
            $a_2$ & $a_4$ & $a_6$ & $a_1$ & $a_5$ & $a_3$ \\
            $a_3$ & $a_2$ & $a_5$ & $a_6$ & $a_1$ & $a_4$ \\
            $a_4$ & $a_3$ & $a_1$ & $a_5$ & $a_6$ & $a_2$ \\
            $a_5$ & $a_1$ & $a_2$ & $a_4$ & $a_3$ & $a_6$ \\
            $a_6$ & $a_5$ & $a_4$ & $a_3$ & $a_2$ & $a_1$
        \end{tabular}
        \caption{An instance with asymmetric preferences}
        \label{table:asymmetric}
    \end{table}
    \captionsetup[table]{name=Table}
    
    The instance admits the stable matching consisting of middle-choice pairs, i.e., 
    $$M_1=\{\{a_1,a_2\},\{a_3,a_6\},\{a_4,a_5\}\},$$ 
    but it also admits the stable matching 
    $$M_2=\{\{a_1,a_3\},\{a_2,a_5\},\{a_4,a_6\}\}.$$
\end{proof}

For \textit{IC}, it is interesting to see that the expected number of stable partitions remains small, averaging fewer than 7 stable partitions and fewer than 3 reduced stable partitions even for $n=500$. Recall that \citet{pittel19} showed that for \emph{IC} instances with $n$ even, with high probability, the number of reduced stable partitions grows in the order of $O(\sqrt{n})$. However, the maximum number of stable partitions observed is 513 for an instance of size $n=500$, whereas the highest observed number of reduced stable partitions for any \emph{IC} instance is 65, admitted by an instance with 300 agents. Thus, although there appear to be only very few stable partitions on average, outlier instances that admit many more such solutions are present and observable. Note that the values for $n$ odd are just slightly below the values for $n$ even, and this is probably caused by the increased likelihood of invariant odd cycles, which leave less flexibility to admit many stable partitions. 

For \textit{2-IC}, the average number of stable partitions and reduced stable partitions is larger and grows slightly more quickly than for \emph{IC} (starting from $n=40$), although remaining at fewer than 27, even for $n=500$. We highlight an outlier at $n=30$ in the otherwise increasing sequences of values -- this is probably due to the fact that $\frac{30}{2}=15$ is odd and therefore there is an increased likelihood of odd cycles, similar to what we observed for \emph{2-IC} in Figure \ref{fig:Pn_subplots}. Another significant difference between \emph{IC} and \emph{2-IC} is that the maximum number of stable partitions observed in \emph{2-IC} instances is 4815 (for an instance with 401 agents), which also admits 300 reduced stable partitions; this is much higher than for \emph{IC}.

As can be seen in Tables \ref{table:numstructureseven} and \ref{table:numstructuresodd}, \emph{Attributes} and \emph{Mallows-Euclidean} admit fewer than 10 stable partitions on average for instances with at most 201 agents, but from thereon we can spot a fast incline, especially for \emph{Mallows-Euclidean}, which admits more than 152 stable partitions on average for instances with 501 agents. 

Unsurprisingly, we observed outlier instances for both cultures that admit a very high number of stable partitions. Multiple \emph{Attributes} instances with $n$ between 300 and 501 that we observed admit $6,561$ stable partitions. While we do not have a theoretical explanation for this, it is interesting to note that $6,561=3^8$ and so it could be, for example, that the stable cycles contain 8 even-length cycles and their decompositions, as well as some (potentially none) other fixed cycles (such as cycles of odd length). Because there are 2 decompositions into stable collections of transpositions for each cycle of even length longer than 2, there are 3 choices for each such cycle: whether to include the long cycle, or one of the collections in the stable partition.

We also observed a \emph{Mallows-Euclidean} instance with 401 agents that admits $10,125=3^45^3$ stable partitions. Interestingly, while the highest number of reduced stable partitions overall was a \emph{Mallows-Euclidean} instance with 401 agents (admitting 432 such structures), \emph{Attributes} instances admitted no more than 256 reduced stable partitions. Similar to previous observations regarding the solvability probability, there is little difference between the observed averages for $n$ even and odd in these two cultures.

\begin{table}[!htb]
    \centering
    \small
    \begin{tabular}{c c c c c c c c c c c c c}
        \toprule
        & & \textbf{10} & \textbf{20} & \textbf{30} & \textbf{40} & \textbf{60} & \textbf{80} & \textbf{100} & \textbf{200} & \textbf{300} & \textbf{400} & \textbf{500} \\
        \midrule
        \multirow{7}{*}{$\vert \mathcal P \vert$} 
        & \textit{IC} & 1.77 & 2.36 & 2.63 & 2.92 & 3.23 & 3.54 & 3.94 & 5.05 & 5.69 & 6.31 & 6.72 \\
        & \textit{2-IC} & 1.05 & 3.07 & 1.50 & 5.33 & 7.05 & 8.79 & 10.18 & 14.72 & 18.17 & 23.74 & 26.87 \\
        & \textit{Symmetric} & 1.00 & 1.00 & 1.00 & 1.00 & 1.00 & 1.00 & 1.00 & 1.00 & 1.00 & 1.00 & 1.00 \\
        & \textit{Asymmetric} & 1.00 & 1.00 & 1.00 & 1.00 & 1.00 & 1.00 & 1.00 & 1.00 & 1.00 & 1.00 & 1.00 \\
        & \textit{Euclidean} & 1.00 & 1.00 & 1.00 & 1.00 & 1.00 & 1.00 & 1.00 & 1.00 & 1.00 & 1.00 & 1.00 \\
        & \textit{Attributes} & 1.05 & 1.16 & 1.23 & 1.32 & 1.53 & 1.89 & 2.13 & 4.77 & 10.78 & 24.82 & 47.38 \\
        & \textit{M-Euclidean} & 1.24 & 1.31 & 1.44 & 1.60 & 1.95 & 2.53 & 3.11 & 9.68 & 24.01 & 61.39 & 126.89 \\
        \midrule
        
        \multirow{7}{*}{$\vert \mathcal{RP} \vert$}
        & \textit{IC} & 1.38 & 1.64 & 1.75 & 1.86 & 1.97 & 2.07 & 2.17 & 2.46 & 2.61 & 2.74 & 2.83 \\
        & \textit{2-IC} & 1.03 & 1.88 & 1.23 & 2.62 & 3.09 & 3.49 & 3.76 & 4.61 & 5.15 & 5.79 & 6.17 \\
        & \textit{Symmetric} & 1.00 & 1.00 & 1.00 & 1.00 & 1.00 & 1.00 & 1.00 & 1.00 & 1.00 & 1.00 & 1.00 \\
        & \textit{Asymmetric} & 1.00 & 1.00 & 1.00 & 1.00 & 1.00 & 1.00 & 1.00 & 1.00 & 1.00 & 1.00 & 1.00 \\
        & \textit{Euclidean} & 1.00 & 1.00 & 1.00 & 1.00 & 1.00 & 1.00 & 1.00 & 1.00 & 1.00 & 1.00 & 1.00 \\
        & \textit{Attributes} & 1.03 & 1.08 & 1.11 & 1.15 & 1.24 & 1.38 & 1.48 & 2.23 & 3.46 & 5.27 & 7.99 \\
        & \textit{M-Euclidean} & 1.12 & 1.15 & 1.20 & 1.28 & 1.41 & 1.60 & 1.79 & 3.17 & 5.17 & 8.65 & 13.78 \\
        \bottomrule
    \end{tabular}
    \caption{Average number of stable partitions and reduced stable partitions for $n$ even}
    \label{table:numstructureseven}
\end{table}

\begin{table}[!htb]
    \centering
    \small
    \begin{tabular}{c c c c c c c c c c c c c}
        \toprule
        & & \textbf{11} & \textbf{21} & \textbf{31} & \textbf{41} & \textbf{61} & \textbf{81} & \textbf{101} & \textbf{201} & \textbf{301} & \textbf{401} & \textbf{501} \\
        \midrule
        \multirow{7}{*}{$\vert \mathcal P \vert$} 
        & \textit{IC} & 1.14 & 1.37 & 1.54 & 1.73 & 2.08 & 2.30 & 2.51 & 3.58 & 4.29 & 4.95 & 5.32 \\
        & \textit{2-IC} & 1.48 & 2.00 & 2.63 & 3.11 & 4.20 & 5.11 & 5.88 & 9.80 & 13.62 & 16.72 & 21.57 \\
        & \textit{Asymmetric} & 1.00 & 1.00 & 1.00 & 1.00 & 1.00 & 1.00 & 1.00 & 1.00 & 1.00 & 1.00 & 1.00 \\
        & \textit{Euclidean} & 1.00 & 1.00 & 1.00 & 1.00 & 1.00 & 1.00 & 1.00 & 1.00 & 1.00 & 1.00 & 1.00 \\
        & \textit{Attributes} & 1.07 & 1.14 & 1.25 & 1.34 & 1.57 & 1.84 & 2.17 & 5.41 & 12.42 & 23.61 & 50.08 \\
        & \textit{M-Euclidean} & 1.15 & 1.30 & 1.46 & 1.66 & 2.00 & 2.67 & 3.21 & 9.51 & 24.03 & 64.65 & 152.16 \\
        \midrule
        
        \multirow{7}{*}{$\vert \mathcal{RP} \vert$} 
        & \textit{IC} & 1.07 & 1.18 & 1.26 & 1.34 & 1.48 & 1.56 & 1.65 & 1.97 & 2.17 & 2.36 & 2.48 \\
        & \textit{2-IC} & 1.24 & 1.47 & 1.73 & 1.90 & 2.23 & 2.50 & 2.70 & 3.58 & 4.22 & 4.64 & 5.27 \\
        & \textit{Asymmetric} & 1.00 & 1.00 & 1.00 & 1.00 & 1.00 & 1.00 & 1.00 & 1.00 & 1.00 & 1.00 & 1.00 \\
        & \textit{Euclidean} & 1.00 & 1.00 & 1.00 & 1.00 & 1.00 & 1.00 & 1.00 & 1.00 & 1.00 & 1.00 & 1.00 \\
        & \textit{Attributes} & 1.03 & 1.07 & 1.12 & 1.16 & 1.26 & 1.37 & 1.49 & 2.33 & 3.58 & 5.33 & 8.10 \\
        & \textit{M-Euclidean} & 1.07 & 1.15 & 1.22 & 1.29 & 1.42 & 1.65 & 1.81 & 3.15 & 5.27 & 9.06 & 15.08 \\
        \bottomrule
    \end{tabular}
    \caption{Average number of stable partitions and reduced stable partitions for $n$ odd}
    \label{table:numstructuresodd}
\end{table}

Now to inspect the growth pattern of $\vert \mathcal{RP}\vert$ visually, Figure \ref{fig:rpgrowth_subplots} shows the average values plotted on a log scale against the number of agents for our different statistical cultures. As can be seen, the growth of $\vert \mathcal{RP}\vert$ actually slows on average for \emph{IC} and \emph{2-IC} instances as $n$ increases. On the other hand, the growth of $\vert \mathcal{RP}\vert$ for \emph{Attributes} and \emph{M-Euclidean} accelerates over time, and the log-scale plot does suggest an exponential growth. 

\begin{figure}[!htb]
    \centering

    % Subplot 1
    \begin{subfigure}{0.49\textwidth}
        \begin{tikzpicture}
            \begin{axis}[
                width=\textwidth,
                height=6cm,
                ymin=1,       % Set the minimum y-axis value
                ymax=16,         % Set the maximum y-axis value
                xmin=0,
                xmax=501,
                xlabel={$n$},
                ylabel={$\vert \mathcal{RP}\vert$ (log scale)},
                log ticks with fixed point,
                ymode=log,
                grid=both,
                grid style={dashed, gray!30},
                cycle list name=color list,
                every axis plot/.append style={thick},
                title={\textit{n} even},
                title style={
                    yshift=-1.5ex  % Adjust to bring title closer to the plot
                },
                legend cell align={left},
                legend style={
                    at={(0,1)}, % Adjusts the legend position
                    anchor=north west,
                    column sep=1ex, % Space between legend columns
                },
            ]
                \addplot[acmDarkBlue, mark=*] table [x=n, y=RP] {data/RP/icEven.txt};
                \addplot[acmGreen, mark=square*] table [x=n, y=RP] {data/RP/2icEven.txt};
                \addplot[acmPink, mark=triangle*] table [x=n, y=RP] {data/RP/attributesEven.txt};
                \addplot[acmOrange, mark=diamond*] table [x=n, y=RP] {data/RP/MEuclEven.txt};
            \addlegendentry{\emph{IC}}
            \addlegendentry{\emph{2-IC}}
            \addlegendentry{\emph{Attributes}}
            \addlegendentry{\emph{M-Euclidean}}
            \end{axis}
        \end{tikzpicture}
    \end{subfigure}
    \hfill
    % Subplot 2
    \begin{subfigure}{0.49\textwidth}
        \begin{tikzpicture}
            \begin{axis}[
                width=\textwidth,
                height=6cm,
                ymin=1,       % Set the minimum y-axis value
                ymax=16,         % Set the maximum y-axis value
                xmin=0,
                xmax=501,
                xlabel={$n$},
                log ticks with fixed point,
                ymode=log,
                grid=both,
                grid style={dashed, gray!30},
                cycle list name=color list,
                every axis plot/.append style={thick},
                title={\textit{n} odd},
                title style={
                    yshift=-1.5ex  % Adjust to bring title closer to the plot
                },
            ]
                \addplot[acmDarkBlue, mark=*] table [x=n, y=RP] {data/RP/icOdd.txt};
                \addplot[acmGreen, mark=square*] table [x=n, y=RP] {data/RP/2icOdd.txt};
                \addplot[acmPink, mark=triangle*] table [x=n, y=RP] {data/RP/attributesOdd.txt};
                \addplot[acmOrange, mark=diamond*] table [x=n, y=RP] {data/RP/MEuclOdd.txt};
            \end{axis}
        \end{tikzpicture}
    \end{subfigure}

    \caption{Growth of $\vert \mathcal{RP}\vert$}
    \label{fig:rpgrowth_subplots}
\end{figure}
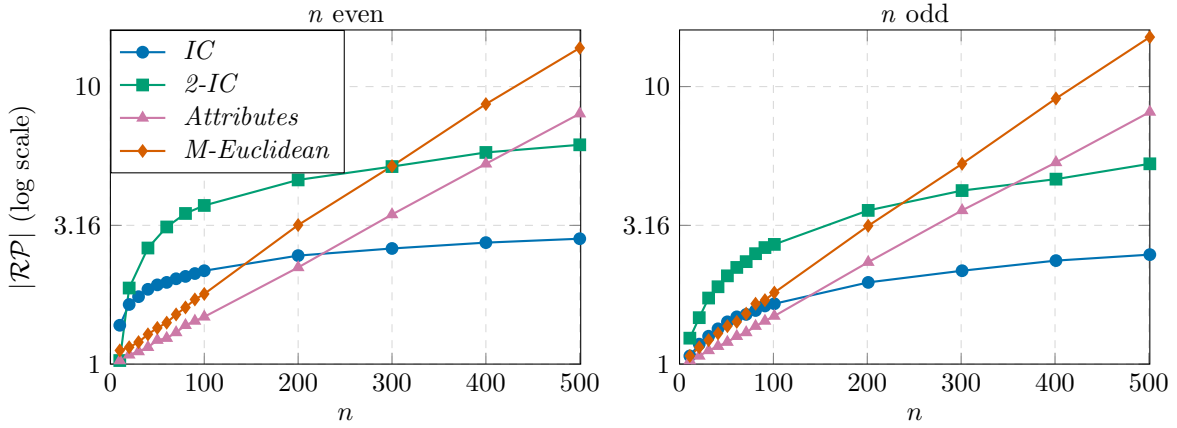

We conclude that, although the number of stable partitions grows large for some cultures, and we could find outlier instances admitting a very large number of stable partitions, most instances only admit a small number of these structures. Furthermore, the number of reduced stable partitions remained small throughout all cultures, and we did not observe any worst-case instances that admit any sort of exponential growth behaviour in this regard. This suggests that enumerating all reduced stable partitions of an instance (within the range of $n$ studied, i.e., $n\leq 501$) is usually feasible in practice. It would be interesting to extend the theoretical $O(\sqrt{n})$ bound (with high probability) on the growth of the number of reduced stable partitions for \emph{IC} instances with an even number of agents by \citet{pittel19} to the other statistical distributions and to an odd number of agents, if possible.

\begin{conj}
    The $O(\sqrt{n})$ bound (with high probability) on the number of reduced stable partitions holds also for \emph{IC} when $n$ is odd, and for \emph{2-IC}, but not for \emph{Attributes} and \emph{Mallows-Euclidean}.
\end{conj}

\subsection{Variability between Stable Partitions}
\label{sec:variability}

Naturally, with many instances admitting a large number of stable partitions, we could also expect the number of distinct cycles contained within them to vary a lot. However, we will see that this is not the case, although the even-length cycles of two stable partitions of an instance could be completely disjoint in general. However, \citet{glitzner2024structuralalgorithmicresultsstable} remarked that despite this, any instance with $n$ agents admits at most $O(n^2)$ different stable cycles, which we extend as follows.

\begin{theorem}
\label{thm:oi}
    Let $I$ be an {\sc sr} instance with $n\geq2$ agents and let $\mathcal O_I$ denote the odd cycles of any stable partition of $I$. If $I$ is solvable, then $\vert\mathcal{O}_I\vert=0$ (if $n$ is even) or $\vert\mathcal{O}_I\vert=1$ (if $n$ is odd). Otherwise, if $I$ is unsolvable, then $\vert\mathcal{O}_I\vert\geq 2$ (if $n$ is even) or $\vert\mathcal{O}_I\vert\geq1$ (if $n$ is odd), and $\vert\mathcal{O}_I\vert\leq\left\lfloor\frac{n}{3}\right\rfloor+((n$ mod $3)$ mod $2)$. Furthermore, these bounds are tight for every $n$.
\end{theorem}
\begin{proof}
    First, note that by Theorem \ref{thm:tan91}, $I$ is unsolvable if and only if $\mathcal{O}_I$ contains at least one cycle of odd length at least 3. Furthermore, any stable partition can contain at most one 1-cycle (due to complete preference lists and stability). Clearly, an even-sized instance (i.e., $n$ is even) must have an even number of odd cycles (or none), and an odd-sized instance must have an odd number of odd cycles. Together, this proves the statements for $I$ solvable and also establishes the lower bounds for $I$ unsolvable.
    
    Tightness of the lower bound for $n$ even can be verified easily by constructing a solvable instance with $n-4$ agents (e.g., an instance with symmetric preferences) and adding four agents forming a 3-cycle and a 1-cycle. For the tightness of the lower bound for $n$ odd, construct a solvable instance with $n-3$ agents and add a 3-cycle.

    For the upper bound, note that if we want to maximise the number of odd cycles in the stable partition of an instance with a given number of agents, it is never beneficial to add cycles of odd length longer than 3. We also said that any stable partition can contain at most one 1-cycle. Thus, any instance with $n$ agents can contain at most $\left\lfloor\frac{n}{3}\right\rfloor$ 3-cycles. Furthermore, if $n$ mod $3=1$, then the remaining agent must form a 1-cycle in the stable partition. However, if $n$ mod $3=2$, then the remaining agents not in 3-cycles must form a transposition. Therefore, the maximum number of odd cycles follows the infinite sequence $1,0,1,2,1,2,3,2,3,\dots$ (for $n=1,2,3,4,\dots$) given by $\left\lfloor\frac{n}{3}\right\rfloor+((n$ mod $3)$ mod $2)$, which corresponds to OEIS sequence A008611 (shifted by 1) \cite{a008611}. 
\end{proof}

Another key result due to \citet{glitzner2024structuralalgorithmicresultsstable} is that the set of reduced stable partitions of any {\sc sr} instance $I$ is in a bijective correspondence with the set of stable matchings of a solvable sub-instance of $I$ (an instance with potentially some agents deleted and potentially some agents made mutually unacceptable, referred to as an {\sc sri} instance due to potentially incomplete preference lists).

\begin{theorem}[\cite{glitzner2024structuralalgorithmicresultsstable}]
\label{thm:glitzner24}
    Let $I$ be an {\sc sr} instance with $n$ agents and let $I_T$ be an {\sc sri} instance with $n'\leq n$ agents derived from $I$ (using a method described in \cite{glitzner2024structuralalgorithmicresultsstable}). Then $M_i$ is a stable matching of $I_T$ if and only if the concatenation of $M_i$ (as transpositions) and the invariant odd cycles $\mathcal{O}_I$ of $I$  forms a reduced stable partition of $I$.
\end{theorem}

We use this result to establish the following bounds.

\begin{theorem}
\label{thm:finegrainedrsc}
    Let $I$ be an {\sc sr} instance with $n\geq2$ agents. Let $\mathcal O_I$ denote the odd cycles of $I$, let $\mathcal{RSC}$ denote the reduced stable cycles of $I$, let $I_T$ be the instance derived from $I$ (as in \cite{glitzner2024structuralalgorithmicresultsstable}, with $n'\leq n$ agents), and finally let $\mathcal{SP}(I_T)$ be the stable pairs of $I_T$. Then, $1\leq \vert\mathcal{RSC}\vert = \vert\mathcal{O}_I\vert+\vert\mathcal{SP}(I_T)\vert \leq \frac{n(n-1)}{2}+1$.
\end{theorem}
\begin{proof}
    The (tight) lower bound holds trivially: clearly any instance with two agents is solvable and admits a unique stable partition consisting of a transposition.

    The equality holds because, by Theorem \ref{thm:glitzner24}, any reduced stable partition consists of stable pairs of $I_T$ (as transpositions) and the invariant odd cycles $\mathcal{O}_I$. Due to the bijective correspondence between stable matchings of $I_T$ and reduced stable partitions of $I$, it follows that $\vert\mathcal{RSC}\vert = \vert\mathcal{O}_I\vert+\vert\mathcal{SP}(I_T)\vert$.

    The upper bound holds because, by Theorem \ref{thm:oi}, $\mathcal{O}_I$ contains at most $\left\lfloor\frac{n-n'}{3}\right\rfloor+((n-n'$ mod 3) mod 2) cycles, and the stable pairs can clearly contain at most $n'$ choose two pairs (i.e., at most $\frac{n'(n'-1)}{2}$ pairs). Therefore, $\vert\mathcal{O}_I\vert+\vert\mathcal{SP}(I_T)\vert \leq \left\lfloor\frac{n-n'}{3}\right\rfloor+((n-n'$ mod 3) mod 2) $+\frac{n'(n'-1)}{2}\leq \frac{n-n'}{3}+1+\frac{n'(n'-1)}{2}\leq \frac{3(n')^2-5n'+2n}{6}+1$, so it is left to show that $3(n')^2-5n'+2n\leq3n^2-3n$, which is indeed the case here because $n'\leq n$ and $n\geq 2$.
\end{proof}

This allows us to provide the following asymptotic upper bounds on the number of stable cycles.

\begin{corollary}
    Let $I$ be an {\sc sr} instance and let $\mathcal{SC}$ and $\mathcal{RSC}$ be the stable cycles and reduced stable cycles of $I$, respectively. Then $\vert\mathcal{RSC}\vert\leq\vert\mathcal{SC}\vert\leq 3\vert\mathcal{RSC}\vert$, therefore $\vert\mathcal{SC}\vert=O(n^2)$.
\end{corollary}
\begin{proof}
    We will argue that any longer even cycle of even length (precisely the cycles in $\mathcal{SC}\setminus\mathcal{RSC}$) can be broken into collections of reduced stable cycles, but that any reduced stable cycle can be in at most two such collections.
    
    By Theorem \ref{thm:tan91}, any cycle $c\in\mathcal{SC}\setminus\mathcal{RSC}$ can be broken into a collection of transpositions $C\subseteq\mathcal{RSC}$. However, as shown in \cite{glitzner2024structuralalgorithmicresultsstable}, any two agents can appear consecutively in at most two cycles of length longer than 2. Therefore, for any three distinct even-length cycles $c_1,c_2,c_3\in\mathcal{SC}\setminus\mathcal{RSC}$ and any choice of transposition decomposition (resulting in collections of transpositions) $C_1,C_2,C_2\subseteq\mathcal{RSC}$, we must have $C_1\cap C_2\cap C_3=\emptyset$. Precisely, any transposition in $\mathcal{RSC}$ can be part of at most two transposition-decomposition collections. Therefore, $\vert\mathcal{SC}\setminus\mathcal{RSC}\vert\leq2\vert\mathcal{RSC}\vert$. Note that $\mathcal{RSC}\subseteq \mathcal{SC}$, so  $\vert\mathcal{SC}\vert=\vert\mathcal{SC}\setminus\mathcal{RSC}\vert+\vert\mathcal{RSC}\vert$. Thus, the stated in equality of $\vert\mathcal{RSC}\vert\leq\vert\mathcal{SC}\vert\leq 3\vert\mathcal{RSC}\vert$ follows.
    
    The asymptotic bound follows from the result we established in Theorem \ref{thm:finegrainedrsc}. Specifically, we have that $\vert\mathcal{RSC}\vert\leq\frac{n(n-1)}{2}+1=O(n^2)$.
\end{proof}

Now with these formal results in mind, lets return to empirical observations about the stable cycles. Tables \ref{table:numcycleseven}-\ref{table:numcyclesodd} show the average number of stable cycles ($\mathcal{SC}$) and reduced stable cycles ($\mathcal{RSC}$) contained within $\mathcal{RP}$ and $\mathcal P$ for the different statistical cultures, and for different even and odd numbers of agents $n$. 

For \textit{IC}, the results suggest that the expected number of stable cycles grows slightly faster than the expected number of reduced stable cycles, averaging a bit higher than the necessary minimum number of cycles of roughly $\frac{n}{2}$ (assuming most cycles are of length 2, as we will see later is typically the case). Interestingly, the maximum number of stable cycles and reduced stable cycles admitted by any instance were 323 and 315, respectively, both admitted by an instance with $n=500$. 

In fact, for the other statistical cultures, it is easily seen in the tables that there are only small differences between the average numbers of cycles each of them admits (with respect to a fixed number of agents). The highest numbers of stable cycles and reduced stable cycles are admitted by \emph{2-IC} instances with $n=500$ (coming to 330 and 323, respectively), whereas the highest such numbers admitted by \emph{Mallows-Euclidean} instances only come to 273 and 265 (for multiple instances with $n=500$ and $n=501)$, respectively.

\begin{table}[!htb]
    \centering
    \footnotesize
    \begin{tabular}{c c c c c c c c c c c c c}
        \toprule
        & & \textbf{10} & \textbf{20} & \textbf{30} & \textbf{40} & \textbf{60} & \textbf{80} & \textbf{100} & \textbf{200} & \textbf{300} & \textbf{400} & \textbf{500} \\
        \midrule
        \multirow{4}{*}{$\vert \mathcal{SC} \vert$}  
        & \textit{IC} & 6.17 & 12.01 & 17.33 & 22.63 & 32.87 & 43.16 & 53.27 & 103.70 & 153.89 & 203.70 & 253.95 \\
        & \textit{2-IC} & 4.80 & 12.30 & 14.22 & 23.85 & 34.70 & 45.23 & 55.53 & 106.29 & 156.80 & 207.09 & 257.30 \\
        & \textit{Attributes} & 5.06 & 10.16 & 15.21 & 20.28 & 30.40 & 40.63 & 50.78 & 101.63 & 152.65 & 203.65 & 254.76 \\
        & \textit{M-Euclidean} & 5.31 & 10.29 & 15.33 & 20.41 & 30.45 & 40.58 & 50.67 & 101.08 & 151.33 & 201.60 & 251.95 \\
        \cmidrule{1-13}
        \multirow{4}{*}{$\vert \mathcal{RSC} \vert$} 
        & \textit{IC} & 5.80 & 11.41 & 16.63 & 21.87 & 32.04 & 42.28 & 52.36 & 102.66 & 152.79 & 202.56 & 252.77 \\
        & \textit{2-IC} & 4.77 & 11.56 & 14.00 & 22.68 & 33.30 & 43.71 & 53.93 & 104.47 & 154.86 & 205.05 & 255.20 \\
        & \textit{Attributes} & 5.04 & 10.09 & 15.11 & 20.13 & 30.19 & 40.30 & 50.38 & 100.81 & 152.65 & 201.94 & 252.58 \\
        & \textit{M-Euclidean} & 5.20 & 10.15 & 15.14 & 20.16 & 30.10 & 40.10 & 50.08 & 99.90 & 149.64 & 199.38 & 249.18 \\
        \bottomrule
    \end{tabular}
    \caption{Average number of cycles in stable partitions and reduced stable partitions for $n$ even}
    \label{table:numcycleseven}
\end{table}

\begin{table}[!htb]
    \centering
    \footnotesize
    \begin{tabular}{c c c c c c c c c c c c c}
        \toprule
        & & \textbf{11} & \textbf{21} & \textbf{31} & \textbf{41} & \textbf{61} & \textbf{81} & \textbf{101} & \textbf{201} & \textbf{301} & \textbf{401} & \textbf{501} \\
        \midrule
        \multirow{4}{*}{$\vert \mathcal{SC} \vert$}  
        & \textit{IC} & 5.24 & 10.04 & 14.93 & 19.97 & 30.05 & 39.98 & 50.19 & 100.33 & 150.47 & 200.98 & 250.72 \\
        & \textit{2-IC} & 6.20 & 11.37 & 16.85 & 21.93 & 32.29 & 42.52 & 52.69 & 103.42 & 153.74 & 203.66 & 254.15 \\
        & \textit{Attributes} & 5.89 & 10.86 & 15.89 & 20.91 & 31.01 & 41.15 & 51.30 & 102.22 & 153.20 & 204.21 & 255.26 \\
        & \textit{M-Euclidean} & 5.83 & 10.90 & 15.94 & 20.97 & 30.98 & 41.16 & 51.17 & 101.51 & 151.88 & 202.29 & 252.67 \\
        \cmidrule{1-13}
        \multirow{4}{*}{$\vert \mathcal{RSC} \vert$} 
        & \textit{IC} & 5.17 & 9.87 & 14.69 & 19.66 & 29.62 & 39.49 & 49.64 & 99.60 & 149.63 & 200.03 & 249.72 \\
        & \textit{2-IC} & 5.97 & 10.94 & 16.21 & 21.17 & 31.35 & 41.45 & 51.53 & 101.94 & 152.10 & 201.94 & 252.29 \\
        & \textit{Attributes} & 5.86 & 10.79 & 15.78 & 20.76 & 30.77 & 40.83 & 50.89 & 101.37 & 151.91 & 202.47 & 253.07 \\
        & \textit{M-Euclidean} & 5.76 & 10.76 & 15.73 & 20.71 & 30.62 & 40.65 & 50.57 & 100.35 & 150.16 & 200.02 & 249.83 \\
        \bottomrule
    \end{tabular}
    \caption{Average number of cycles in stable partitions and reduced stable partitions for $n$ odd}
    \label{table:numcyclesodd}
\end{table}

It is surprising to observe these small numbers of distinct cycles, but this does not contradict anything known about the structure of stable partitions per se. The combinatorial nature of choosing (the right) stable cycles to build a partition allows even a small number of distinct cycles to lead to a much larger number of stable partitions that they can form.

\subsection{Stable Matchings and Pairs}
\label{sec:nummatchings}

Taking a slight detour from unsolvable instances, we turn to the question of whether the observations above for stable partitions also hold for stable matchings (denoted by $\mathcal M$) and the stable pairs (denoted by $\mathcal{SP}$) contained within them with respect to solvable {\sc sr} instances. Recall that stable matchings and reduced stable partitions are equivalent in this case. Note that for consistency with $\mathcal{RSC}$, we include the fixed point present in solvable instances containing an odd number of agents in $\vert \mathcal{SP}\vert$ (in the sense that this single agent is paired to itself).

Tables \ref{table:nummatchingspairseven}-\ref{table:nummatchingspairsodd} show the average values we observed for the four relevant statistical cultures and some selected values of $n$. N/A indicates the case that there were no solvable instances to average over. We can see that the number of stable matchings and stable cycles admitted by solvable instances is, on average, just slightly higher than the number of reduced stable partitions and their cycles averaged over both solvable and unsolvable instances together. This makes sense, as we need to take into account the presence of invariant odd cycles, which decreases the potential for a larger number of reduced stable partitions (informally, anything that is fixed or invariant cannot vary between stable partitions). There are a few observations which seem out of place, such as an average of two stable matchings for \emph{Mallows-Euclidean} instances with $n=300$; this is likely because there were few such solvable instances to average over.

\begin{table}[!htb]
    \centering
    \footnotesize
    \begin{tabular}{c c c c c c c c c c c c c}
        \toprule
        & & \textbf{10} & \textbf{20} & \textbf{30} & \textbf{40} & \textbf{60} & \textbf{80} & \textbf{100} & \textbf{200} & \textbf{300} & \textbf{400} & \textbf{500} \\
        \midrule        
        \multirow{4}{*}{$\vert \mathcal M \vert$} 
        & \textit{IC} & 1.42 & 1.76 & 1.91 & 2.08 & 2.26 & 2.40 & 2.56 & 3.05 & 3.30 & 3.53 & 3.58 \\
        & \textit{2-IC} & 1.05 & 2.02 & 1.28 & 3.03 & 3.77 & 4.36 & 4.84 & 6.57 & 7.49 & 8.88 & 9.45 \\
        & \textit{Attributes} & 1.03 & 1.09 & 1.13 & 1.17 & 1.28 & 1.41 & 1.54 & 2.42 & 3.74 & 5.85 & 8.93 \\
        & \textit{M-Euclidean} & 1.14 & 1.18 & 1.24 & 1.31 & 1.49 & 1.75 & 1.93 & 2.57 & 2.00 & 4.00 & N/A \\
        \cmidrule{1-13}
        \multirow{4}{*}{$\vert \mathcal{SP} \vert$} 
        & \textit{IC} & 5.97 & 11.96 & 17.50 & 23.13 & 33.94 & 44.63 & 55.18 & 107.52 & 159.12 & 210.48 & 261.75 \\
        & \textit{2-IC} & 5.00 & 11.94 & 15.58 & 23.82 & 35.19 & 46.19 & 57.13 & 110.54 & 162.66 & 214.84 & 266.29 \\
        & \textit{Attributes} & 5.06 & 10.18 & 15.25 & 20.35 & 30.54 & 40.74 & 50.97 & 101.90 & 152.97 & 204.05 & 255.33 \\
        & \textit{M-Euclidean} & 5.30 & 10.35 & 15.46 & 20.59 & 30.86 & 41.20 & 51.37 & 101.93 & 152.00 & 204.00 & N/A \\
        \bottomrule
    \end{tabular}
    \caption{Average number of stable matchings and pairs of solvable instances for $n$ even}
    \label{table:nummatchingspairseven}
\end{table}

\begin{table}[!htb]
    \centering
    \footnotesize
    \begin{tabular}{c c c c c c c c c c c c c}
        \toprule
        & & \textbf{11} & \textbf{21} & \textbf{31} & \textbf{41} & \textbf{61} & \textbf{81} & \textbf{101} & \textbf{201} & \textbf{301} & \textbf{401} & \textbf{501} \\
        \midrule        
        \multirow{4}{*}{$\vert \mathcal M \vert$} 
        & \textit{IC} & 1.10 & 1.26 & 1.30 & 1.40 & 1.55 & 1.55 & 1.83 & N/A & N/A & N/A & N/A \\
        & \textit{2-IC} & 1.27 & 1.57 & 1.90 & 2.15 & 2.62 & 2.98 & 3.28 & 2.80 & 4.00 & N/A & N/A \\
        & \textit{Attributes} & 1.04 & 1.08 & 1.13 & 1.17 & 1.29 & 1.41 & 1.53 & 2.43 & 3.88 & 6.21 & 8.33 \\
        & \textit{M-Euclidean} & 1.09 & 1.16 & 1.25 & 1.34 & 1.52 & 1.77 & 1.86 & 3.14 & 2.00 & 4.00 & N/A \\
        \cmidrule{1-13}
        \multirow{4}{*}{$\vert \mathcal{SP} \vert$} 
        & \textit{IC} & 6.21 & 11.62 & 16.71 & 22.00 & 32.48 & 42.30 & 53.83 & N/A & N/A & N/A & N/A \\
        & \textit{2-IC} & 6.57 & 12.19 & 17.95 & 23.49 & 34.45 & 45.43 & 55.78 & 107.40 & 157.00 & N/A & N/A \\
        & \textit{Attributes} & 6.08 & 11.15 & 16.28 & 21.34 & 31.57 & 41.76 & 51.94 & 102.93 & 154.14 & 205.13 & 255.92 \\
        & \textit{M-Euclidean} & 6.18 & 11.33 & 16.48 & 21.64 & 31.89 & 42.22 & 52.43 & 103.32 & 152.67 & 205.00 & N/A \\
        \bottomrule
    \end{tabular}
    \caption{Average number of stable matchings and pairs of solvable instances for $n$ odd}
    \label{table:nummatchingspairsodd}
\end{table}

Although this was already suspected after observing the number of reduced stable partitions in Section \ref{sec:numpartitions}, this is further evidence suggesting that we should expect the number of reduced stable partitions and stable matchings to be small in most cases.

\subsection{Odd Cycles}
\label{sec:oddcycles}

Returning solely to unsolvable instances now, we already know that cycles of odd length are of particular importance, making an instance $I$ unsolvable and being invariant between all stable partitions of $I$. However, how many and which cycles of odd length should we expect on average? Table \ref{table:oddvagents} gives some indications for the numbers and lengths to expect. 

We can observe an interesting tradeoff between the two extreme cultures \emph{IC} and \emph{Mallows-Euclidean}, where the former shows a significant growth in its average odd cycle length between $n=10$ and $n=500$, whereas the latter shows a much smaller (or no) growth, while the opposite is the case when inspecting their average numbers of odd cycles. 

For \textit{IC}, we can see that, in expectation, an instance only admits a very small number of odd cycles, and this parameter grows very slowly. This is consistent with previous results by \citet{pittel19}. Notice that any even-sized unsolvable instance must admit at least two odd cycles, so the results indicate that most instances do not admit any more than those two. Note also that the observations are in line with the bound proven by \citet{pittel19} for $n$ even (the bound states that, in expectation, the total number of cycles of odd length grows in the order of at most $n^{\frac{1}{4}}\log n$). 

The \emph{Attributes} and \emph{Mallows-Euclidean} instances only admit very short odd cycles on average, not much longer than 3 even for $n=501$, but it is interesting to see a small dip for $n=101$ in between the values for $n=11$ and $n=501$. The average number of odd cycles also remains low for \emph{2-IC} and \emph{Attributes}, with neither admitting even 5 such cycles on average. We will see later that we did not observe any instance with at most 101 agents from any statistical culture that admits more than 11 odd-length cycles (although it is simple to construct an instance with $n=99$ that admits 33 cycles of length 3 in its stable partition).

\begin{table}[!htb]
    \centering
    \begin{tabular}{c c c c c c c c}
        \toprule
        & & \multicolumn{3}{c}{$n$ even} & \multicolumn{3}{c}{$n$ odd} \\
        \cmidrule(lr){3-5} \cmidrule(lr){6-8}
        
        & Culture & \textbf{10} & \textbf{100} & \textbf{500} & \textbf{11} & \textbf{101} & \textbf{501} \\
        \midrule
        \multirow{4}{*}{\text{Average odd cycle length}} 
        & \textit{IC} & 2.62 & 6.08 & 10.51 & 3.84 & 7.14 & 11.66 \\
        & \textit{2-IC} & 2.43 & 4.82 & 8.20 & 3.11 & 5.51 & 8.83 \\
        & \textit{Attributes} & 2.11 & 2.62 & 3.00 & 3.02 & 2.86 & 3.01 \\
        & \textit{M-Euclidean} & 2.25 & 2.78 & 3.04 & 3.16 & 2.82 & 3.04 \\
        \midrule
        \multirow{4}{*}{\text{Average number of odd cycles}} 
        & \textit{IC} & 2.00 & 2.04 & 2.15 & 1.00 & 1.12 & 1.38 \\
        & \textit{2-IC} & 2.00 & 2.32 & 2.74 & 1.14 & 1.63 & 2.24 \\
        & \textit{Attributes} & 2.00 & 2.31 & 4.44 & 1.09 & 1.99 & 4.43 \\
        & \textit{M-Euclidean} & 2.00 & 3.20 & 12.58 & 1.09 & 3.19 & 12.62 \\
        \bottomrule
    \end{tabular}
    \caption{Cycle properties of unsolvable instances and  different $n$}
    \label{table:oddvagents}
\end{table}

We now turn to the number of agents in odd cycles (also denoted by $n_{odd}$ in the literature \cite{mertens05}), as shown in Figure \ref{fig:noddic} for \textit{IC} instances. For $n$ even, the results roughly match the following conjecture by \citet{mertens05}: $n_{odd} \simeq 2.38 \sqrt{\frac{n}{\ln n}}$. Note that the conjecture is based on a best-fit estimate for a different range of $n$ than plotted here, which explains the difference between our observations and the estimate. Again, we complement this result by fitting a similar function to our observations for $n$ odd and conjecture the following.

\begin{conj}
    For \emph{IC} instances with an odd number of agents, $n_{odd} \simeq 1.74 \sqrt{\frac{n}{\ln n}}$.
\end{conj}

Both conjectures (ours and the one by \citet{mertens05}) are also plotted in Figure \ref{fig:noddic}.

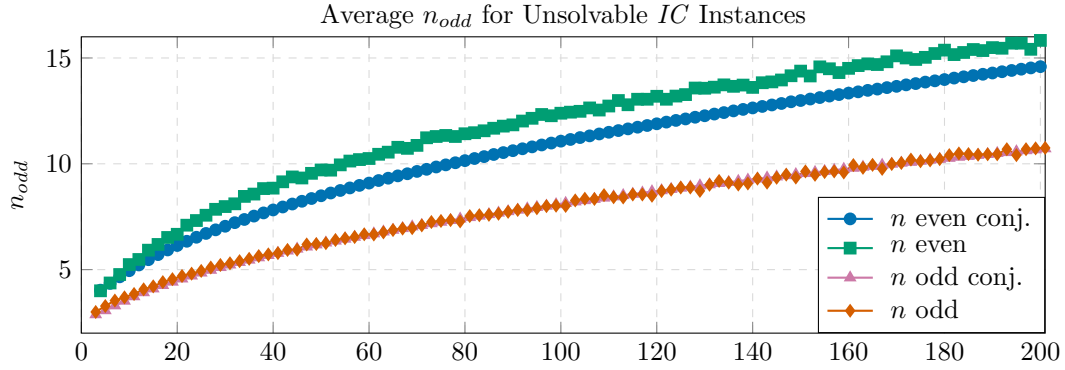
\begin{figure}[!htb]
    \centering
    \begin{tikzpicture}
        \begin{axis}[
            width=.9\textwidth,
            height=5.5cm,
            ylabel={$n_{odd}$},
            ymin=2,       % Set the minimum y-axis value
            ymax=16,         % Set the maximum y-axis value
            xmin=0,
            xmax=201,
            grid=both,
            grid style={dashed, gray!30},
            cycle list name=color list,
            every axis plot/.append style={thick},
            title={Average $n_{odd}$ for Unsolvable \textit{IC} Instances},
            title style={
                yshift=-1.5ex  % Adjust to bring title closer to the plot
            },
            legend cell align={left},
            legend style={
                at={(1,0)}, % Adjusts the legend position
                anchor=south east,
                column sep=1ex, % Space between legend columns
            },
        ]
        \addplot[acmDarkBlue, mark=*] table [x=n, y=evenConj] {data/nodd/icEven.txt};
        \addplot[acmGreen, mark=square*] table [x=n, y=nodd] {data/nodd/icEven.txt};
        \addplot[acmPink, mark=triangle*] table [x=n, y=oddConj] {data/nodd/icOdd.txt};
        \addplot[acmOrange, mark=diamond*] table [x=n, y=nodd] {data/nodd/icOdd.txt};
        \addlegendentry{$n$ even conj.}
        \addlegendentry{$n$ even}
        \addlegendentry{$n$ odd conj.}
        \addlegendentry{$n$ odd}
        \end{axis}
    \end{tikzpicture}
    \caption{Expected $n_{odd}$ in \emph{IC} instances}
    \label{fig:noddic}
\end{figure}

Furthermore, Table \ref{table:nodddiverse} gives some interesting concrete numbers for $n_{odd}$ for the four relevant statistical cultures and for $n$ even and odd. We see moderate growth for all cultures, where \textit{Mallows-Euclidean} grows especially fast and \textit{Attributes} the slowest. On average, though, there are fewer than 39 agents in odd-length cycles even in the former culture and with 501 agents.

\begin{table}[!htb]
    \centering
    \begin{tabular}{c c c c c c c c}
        \toprule
        & & \multicolumn{3}{c}{$n$ even} & \multicolumn{3}{c}{$n$ odd} \\
        \cmidrule(lr){3-5} \cmidrule(lr){6-8}
        & Culture & \textbf{10} & \textbf{100} & \textbf{500} & \textbf{11} & \textbf{101} & \textbf{501} \\
        \midrule
        \multirow{4}{*}{\text{Average $n_{odd}$}} 
        & \textit{IC} & 5.24 & 12.39 & 22.56 & 3.85 & 8.01 & 16.13 \\
        & \textit{2-IC} & 4.86 & 11.20 & 22.45 & 3.56 & 8.96 & 19.75 \\
        & \textit{Attributes} & 4.23 & 6.03 & 13.30 & 3.30 & 5.70 & 13.32 \\
        & \textit{M-Euclidean} & 4.50 & 8.90 & 38.30 & 3.43 & 8.98 & 38.40 \\
        \bottomrule
    \end{tabular}
    \caption{Number of agents in odd-length cycles for different $n$}
    \label{table:nodddiverse}
\end{table}

In the final experiment that focuses specifically on the odd cycles, we investigated the expected number of stable cycles of small odd lengths for unsolvable instances. This answers the following question for small values of $k$: \emph{given an unsolvable IC instance with 100 agents, how many cycles of (odd) length k do we expect?} Figure \ref{fig:cyclesvn_subplots} gives an extensive overview of our findings for different cultures.

\begin{figure}[!htb]
    \centering

    % Subplot 1
    \begin{subfigure}{0.49\textwidth}
        \begin{tikzpicture}
            \begin{axis}[
                width=\textwidth,
                height=4cm,
                ylabel={Number of Cycles},
                ymin=0,       % Set the minimum y-axis value
                ymax=1.3,         % Set the maximum y-axis value
                xmin=0,
                xmax=201,
                grid=both,
                grid style={dashed, gray!30},
                cycle list name=color list,
                every axis plot/.append style={thick},
                title={\textit{IC (n even)}},
                title style={
                    yshift=-1.5ex  % Adjust to bring title closer to the plot
                },
                legend cell align={left},
                legend style={
                    at={(1.42,2.3)}, % Adjusts the legend position
                    anchor=north east,
                    column sep=1ex, % Space between legend columns
                },
            ]
                \addplot[acmDarkBlue, mark=*] table [x=n, y=1-cycles] {data/cycleLengths/icEven.txt};
                \addplot[acmGreen, mark=square*] table [x=n, y=3-cycles] {data/cycleLengths/icEven.txt};
                \addplot[acmPink, mark=triangle*] table [x=n, y=5-cycles] {data/cycleLengths/icEven.txt};
                \addplot[acmOrange, mark=diamond*] table [x=n, y=7-cycles] {data/cycleLengths/icEven.txt};
                \addplot[acmYellow, mark=o] table [x=n, y=9-cycles] {data/cycleLengths/icEven.txt};
                \addplot[acmLightBlue, mark=star] table [x=n, y=11-cycles] {data/cycleLengths/icEven.txt};
            \addlegendentry{Cycles of Length 1}
            \addlegendentry{Cycles of Length 3}
            \addlegendentry{Cycles of Length 5}
            \addlegendentry{Cycles of Length 7}
            \addlegendentry{Cycles of Length 9}
            \addlegendentry{Cycles of Length 11}
            \end{axis}
        \end{tikzpicture}
    \end{subfigure}
    \hfill
    % Subplot 2
    \begin{subfigure}{0.49\textwidth}
        \begin{tikzpicture}
            \begin{axis}[
                width=\textwidth,
                height=4cm,
                ymin=0,       % Set the minimum y-axis value
                ymax=1.1,         % Set the maximum y-axis value
                xmin=0,
                xmax=201,
                grid=both,
                grid style={dashed, gray!30},
                cycle list name=color list,
                every axis plot/.append style={thick},
                title={\textit{IC (n odd)}},
                title style={
                    yshift=-1.5ex  % Adjust to bring title closer to the plot
                },
            ]
                \addplot[acmDarkBlue, mark=*] table [x=n, y=1-cycles] {data/cycleLengths/icOdd.txt};
                \addplot[acmGreen, mark=square*] table [x=n, y=3-cycles] {data/cycleLengths/icOdd.txt};
                \addplot[acmPink, mark=triangle*] table [x=n, y=5-cycles] {data/cycleLengths/icOdd.txt};
                \addplot[acmOrange, mark=diamond*] table [x=n, y=7-cycles] {data/cycleLengths/icOdd.txt};
                \addplot[acmYellow, mark=o] table [x=n, y=9-cycles] {data/cycleLengths/icOdd.txt};
                \addplot[acmLightBlue, mark=star] table [x=n, y=11-cycles] {data/cycleLengths/icOdd.txt};
            \end{axis}
        \end{tikzpicture}
    \end{subfigure}

    \vspace{0.1cm}

    % Subplot 3
    \begin{subfigure}{0.49\textwidth}
        \begin{tikzpicture}
            \begin{axis}[
                width=\textwidth,
                height=4cm,
                ylabel={Number of Cycles},
                ymin=0,       % Set the minimum y-axis value
                ymax=1.4,         % Set the maximum y-axis value
                xmin=0,
                xmax=201,
                grid=both,
                grid style={dashed, gray!30},
                cycle list name=color list,
                every axis plot/.append style={thick},
                title={\textit{2-IC (n even)}},
                title style={
                    yshift=-1.5ex  % Adjust to bring title closer to the plot
                },
            ]
                \addplot[acmDarkBlue, mark=*] table [x=n, y=1-cycles] {data/cycleLengths/2icEven.txt};
                \addplot[acmGreen, mark=square*] table [x=n, y=3-cycles] {data/cycleLengths/2icEven.txt};
                \addplot[acmPink, mark=triangle*] table [x=n, y=5-cycles] {data/cycleLengths/2icEven.txt};
                \addplot[acmOrange, mark=diamond*] table [x=n, y=7-cycles] {data/cycleLengths/2icEven.txt};
                \addplot[acmYellow, mark=o] table [x=n, y=9-cycles] {data/cycleLengths/2icEven.txt};
                \addplot[acmLightBlue, mark=star] table [x=n, y=11-cycles] {data/cycleLengths/2icEven.txt};
            \end{axis}
        \end{tikzpicture}
    \end{subfigure}
    \hfill
    % Subplot 4
    \begin{subfigure}{0.49\textwidth}
        \begin{tikzpicture}
            \begin{axis}[
                width=\textwidth,
                height=4cm,
                ymin=0,       % Set the minimum y-axis value
                ymax=1.1,         % Set the maximum y-axis value
                xmin=0,
                xmax=201,
                grid=both,
                grid style={dashed, gray!30},
                cycle list name=color list,
                every axis plot/.append style={thick},
                title={\textit{2-IC (n odd)}},
                title style={
                    yshift=-1.5ex  % Adjust to bring title closer to the plot
                },
            ]
                \addplot[acmDarkBlue, mark=*] table [x=n, y=1-cycles] {data/cycleLengths/2icOdd.txt};
                \addplot[acmGreen, mark=square*] table [x=n, y=3-cycles] {data/cycleLengths/2icOdd.txt};
                \addplot[acmPink, mark=triangle*] table [x=n, y=5-cycles] {data/cycleLengths/2icOdd.txt};
                \addplot[acmOrange, mark=diamond*] table [x=n, y=7-cycles] {data/cycleLengths/2icOdd.txt};
                \addplot[acmYellow, mark=o] table [x=n, y=9-cycles] {data/cycleLengths/2icOdd.txt};
                \addplot[acmLightBlue, mark=star] table [x=n, y=11-cycles] {data/cycleLengths/2icOdd.txt};
            \end{axis}
        \end{tikzpicture}
    \end{subfigure}

    \vspace{0.1cm}

    % Subplot 5
    \begin{subfigure}{0.49\textwidth}
        \begin{tikzpicture}
            \begin{axis}[
                width=\textwidth,
                height=4cm,
                ylabel={Number of Cycles},
                ymin=0,       % Set the minimum y-axis value
                ymax=2.1,         % Set the maximum y-axis value
                xmin=0,
                xmax=201,
                grid=both,
                grid style={dashed, gray!30},
                cycle list name=color list,
                every axis plot/.append style={thick},
                title={\textit{Attributes (n even)}},
                title style={
                    yshift=-1.5ex  % Adjust to bring title closer to the plot
                },
            ]
                \addplot[acmDarkBlue, mark=*] table [x=n, y=1-cycles] {data/cycleLengths/attributesEven.txt};
                \addplot[acmGreen, mark=square*] table [x=n, y=3-cycles] {data/cycleLengths/attributesEven.txt};
                \addplot[acmPink, mark=triangle*] table [x=n, y=5-cycles] {data/cycleLengths/attributesEven.txt};
                \addplot[acmOrange, mark=diamond*] table [x=n, y=7-cycles] {data/cycleLengths/attributesEven.txt};
                \addplot[acmYellow, mark=o] table [x=n, y=9-cycles] {data/cycleLengths/attributesEven.txt};
                \addplot[acmLightBlue, mark=star] table [x=n, y=11-cycles] {data/cycleLengths/attributesEven.txt};
            \end{axis}
        \end{tikzpicture}
    \end{subfigure}
    \hfill
    % Subplot 6
    \begin{subfigure}{0.49\textwidth}
        \begin{tikzpicture}
            \begin{axis}[
                width=\textwidth,
                height=4cm,
                ymin=0,       % Set the minimum y-axis value
                ymax=2.1,         % Set the maximum y-axis value
                xmin=0,
                xmax=201,
                grid=both,
                grid style={dashed, gray!30},
                cycle list name=color list,
                every axis plot/.append style={thick},
                title={\textit{Attributes (n odd)}},
                title style={
                    yshift=-1.5ex  % Adjust to bring title closer to the plot
                },
            ]
                \addplot[acmDarkBlue, mark=*] table [x=n, y=1-cycles] {data/cycleLengths/attributesOdd.txt};
                \addplot[acmGreen, mark=square*] table [x=n, y=3-cycles] {data/cycleLengths/attributesOdd.txt};
                \addplot[acmPink, mark=triangle*] table [x=n, y=5-cycles] {data/cycleLengths/attributesOdd.txt};
                \addplot[acmOrange, mark=diamond*] table [x=n, y=7-cycles] {data/cycleLengths/attributesOdd.txt};
                \addplot[acmYellow, mark=o] table [x=n, y=9-cycles] {data/cycleLengths/attributesOdd.txt};
                \addplot[acmLightBlue, mark=star] table [x=n, y=11-cycles] {data/cycleLengths/attributesOdd.txt};
            \end{axis}
        \end{tikzpicture}
    \end{subfigure}

    \vspace{0.1cm}

    % Subplot 7
    \begin{subfigure}{0.49\textwidth}
        \begin{tikzpicture}
            \begin{axis}[
                width=\textwidth,
                height=4cm,
                xlabel={Number of Agents ($n$)},
                ylabel={Number of Cycles},
                ymin=0,       % Set the minimum y-axis value
                ymax=4.8,         % Set the maximum y-axis value
                xmin=0,
                xmax=201,
                ytick={0,0.5,2,4},
                grid=both,
                grid style={dashed, gray!30},
                cycle list name=color list,
                every axis plot/.append style={thick},
                title={\textit{M-Euclidean (n even)}},
                title style={
                    yshift=-1.5ex  % Adjust to bring title closer to the plot
                },
            ]
                \addplot[acmDarkBlue, mark=*] table [x=n, y=1-cycles] {data/cycleLengths/MEuclEven.txt};
                \addplot[acmGreen, mark=square*] table [x=n, y=3-cycles] {data/cycleLengths/MEuclEven.txt};
                \addplot[acmPink, mark=triangle*] table [x=n, y=5-cycles] {data/cycleLengths/MEuclEven.txt};
                \addplot[acmOrange, mark=diamond*] table [x=n, y=7-cycles] {data/cycleLengths/MEuclEven.txt};
                \addplot[acmYellow, mark=o] table [x=n, y=9-cycles] {data/cycleLengths/MEuclEven.txt};
                \addplot[acmLightBlue, mark=star] table [x=n, y=11-cycles] {data/cycleLengths/MEuclEven.txt};
            \end{axis}
        \end{tikzpicture}
    \end{subfigure}
    \hfill
    % Subplot 8
    \begin{subfigure}{0.49\textwidth}
        \begin{tikzpicture}
            \begin{axis}[
                width=\textwidth,
                height=4cm,
                xlabel={Number of Agents ($n$)},
                ymin=0,       % Set the minimum y-axis value
                ymax=4.8,         % Set the maximum y-axis value
                xmin=0,
                xmax=201,
                ytick={0,0.5,2,4},
                grid=both,
                grid style={dashed, gray!30},
                cycle list name=color list,
                every axis plot/.append style={thick},
                title={\textit{M-Euclidean (n odd)}},
                title style={
                    yshift=-1.5ex  % Adjust to bring title closer to the plot
                },
            ]
                \addplot[acmDarkBlue, mark=*] table [x=n, y=1-cycles] {data/cycleLengths/MEuclOdd.txt};
                \addplot[acmGreen, mark=square*] table [x=n, y=3-cycles] {data/cycleLengths/MEuclOdd.txt};
                \addplot[acmPink, mark=triangle*] table [x=n, y=5-cycles] {data/cycleLengths/MEuclOdd.txt};
                \addplot[acmOrange, mark=diamond*] table [x=n, y=7-cycles] {data/cycleLengths/MEuclOdd.txt};
                \addplot[acmYellow, mark=o] table [x=n, y=9-cycles] {data/cycleLengths/MEuclOdd.txt};
                \addplot[acmLightBlue, mark=star] table [x=n, y=11-cycles] {data/cycleLengths/MEuclOdd.txt};
            \end{axis}
        \end{tikzpicture}
    \end{subfigure}
    \caption{Expected number of odd-length cycles in unsolvable instances by cycle length}
    \label{fig:cyclesvn_subplots}
\end{figure}
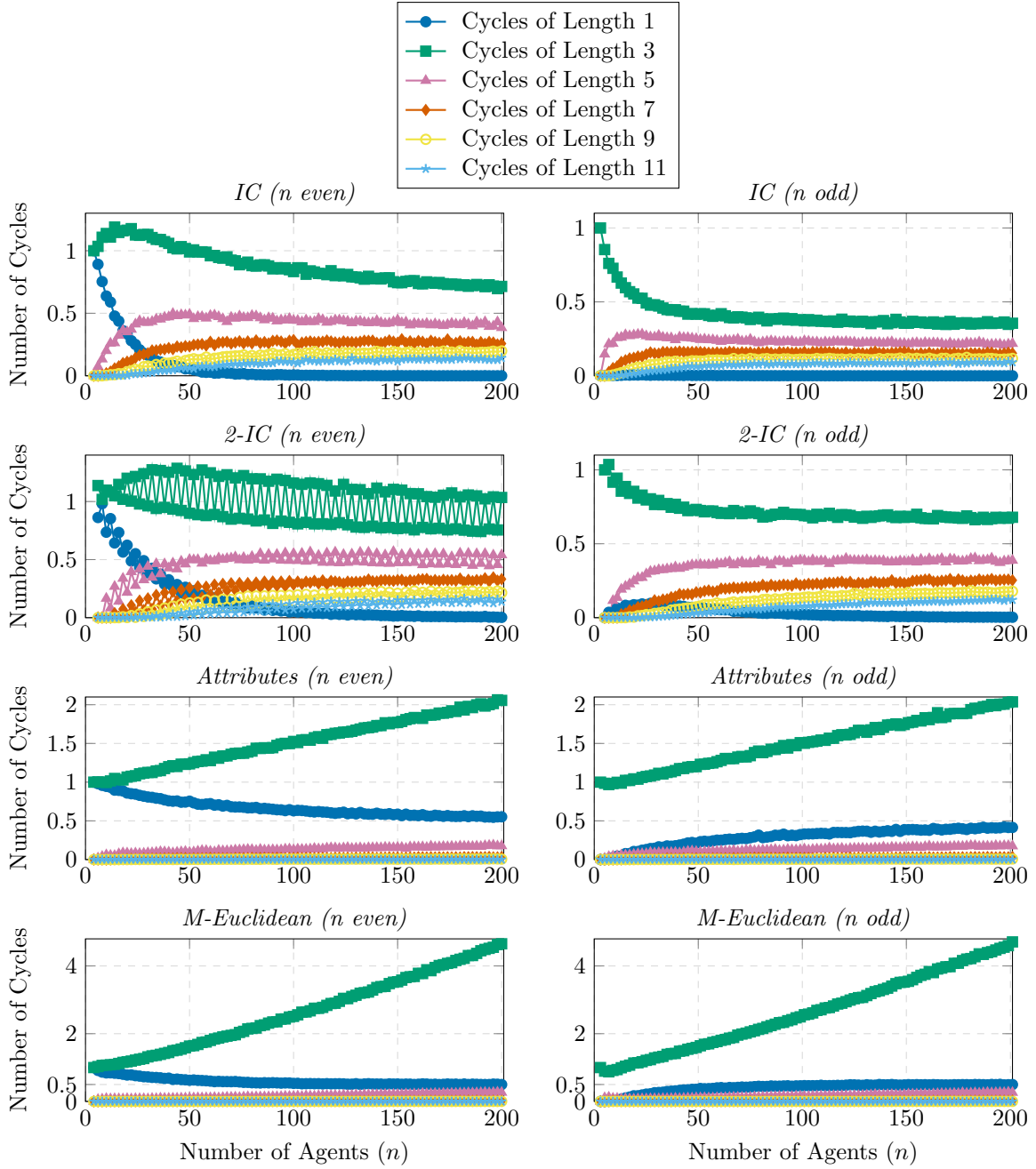

It is clear that for $n=4$, every unsolvable instance admits exactly one stable cycle of length 1 (i.e., a 1-cycle) and one stable cycle of length 3 (a 3-cycle). However, it is interesting to see that for \emph{IC} and $n$ even, while the likelihood of a 1-cycle decreases quickly and is essentially 0 for $n=100$ (this confirms Pittel's various theoretical results, e.g., the super-polynomially high probability that an instance with preferences chosen uniformly at random does not admit a fixed point \cite{pittel19}), the 3-cycles remain the most common odd cycles throughout. Furthermore, the expected numbers of cycles of length 5 or longer remain relatively low compared to that of 3-cycles, although we can see a clear hierarchy in the sense that we should expect more 5-cycles than 7-cycles than 9-cycles, etc. 

Notice that the other cultures also show a decay in the likelihood of 1-cycles when $n$ is even, although not nearly as steep. For example, on average we expect that every second \emph{Mallows-Euclidean} instance containing 200 agents admits a 1-cycle. The situation is even more different when $n$ is odd -- here, \emph{Attributes} and \emph{Mallows-Euclidean} actually show an increase in the likelihood of 1-cycles as the number of agents increases (while this likelihood remains near zero for \emph{IC}). 

We further observe that the alternating jumps in the solvability of \emph{2-IC} instances are reflected in the likelihoods of odd-length cycles, and \emph{Mallows-Euclidean} instances are very likely to admit multiple 3-cycles on average (but few cycles of other odd lengths). A last observation is that for \emph{IC} and $n$ even, the probability of an instance being unsolvable roughly coincides with the expected number of cycles of length 3. This is not the case for $n$ odd, where $1-\hat{P}_n$ is much higher than the expected number of odd cycles of any length.

Although we looked at many different types of observations for the odd cycles in random instances, clear patterns emerged. Specifically, we saw that we usually expect a very small number of odd cycles and their lengths to be short. The likelihood of specific cycle lengths varies by statistical culture, but clear hierarchies are maintained and 3-cycles are by far the most likely.

\section{Two Practical Implications for {\sc Stable Roommates} Problems}
\label{sec:implications}

Before concluding this paper, we give two practical implications of our empirical structural results for {\sc sr} solutions. First, we will establish that maximum stable matchings are very close to being maximum matchings, and then we will discuss what our results for the number of stable solutions mean for NP-hard stable matching and stable partition problems. Note that all values are derived from the same instances and experiments discussed in Section \ref{sec:experiments}.

\paragraph{The Sizes of Maximum Stable Matchings} Our results suggest that while the average number of agents in odd cycles can grow moderately fast as $n$ increases, the average number of odd cycles only exhibits very slow growth with respect to $n$ for most statistical cultures. This gives further motivation for maximum stable matchings as a solution concept, in which all but one agent from each odd cycle will be matched in every such solution and which can be computed efficiently in $O(n^2)$ time \cite{tan91_2}.

To measure how large we expect maximum stable matchings to be, we introduce the following notion as a ratio between the size of a maximum stable matching and a complete matching. 

\begin{definition}[$\alpha(I)$ and $\alpha_n$]
    Let $I$ be an {\sc sr} instance with $n$ agents and let $M$ be a maximum stable matching of $I$. Furthermore, let $M'$ be a maximum matching (not necessarily stable) of $I$, which is always of size $\lfloor \frac{n}{2}\rfloor$ (due to complete preference lists). Now let $\alpha(I)$ denote the ratio between $\vert M\vert$ and $\vert M'\vert$, i.e., $\alpha = \frac{\vert M\vert}{\vert M'\vert}=\frac{\vert M\vert}{\lfloor \frac{n}{2}\rfloor}$. Notice that $\alpha(I)$ is independent of the specific choice of $M$, as all maximum stable matchings have the same size. Finally, for a random {\sc sr} instance with $n$ agents, let $\alpha_n$ denote the expected value of $\alpha(I)$.
\end{definition}

Now, to get a better understanding of the behaviour of this multiplicative ratio and its expected behaviour for our instances, we calculated estimates for $\alpha_n$ (denoted by $\hat\alpha_n$) explicitly for our four statistical cultures of interest for $n$ even and odd, the results of which can be found in Figure \ref{fig:maximumstablevperfect_subplots}. 

As the $y$-axis scaling already indicates, $\alpha_n$ is expected to be high, especially for $n$ odd. Furthermore, all cases suggest an (eventually) increasing trend where $\alpha$ is proportional to $n$, except for odd-sized \emph{Mallows-Euclidean} instances. In fact, $\alpha_n$ appears to be bounded below (in expectation) by 0.99 for \emph{IC, 2-IC} and \emph{Attributes} instances for sufficiently large $n$ ($n\geq 58$, $n\geq 200$ and $n\geq 400$, respectively).

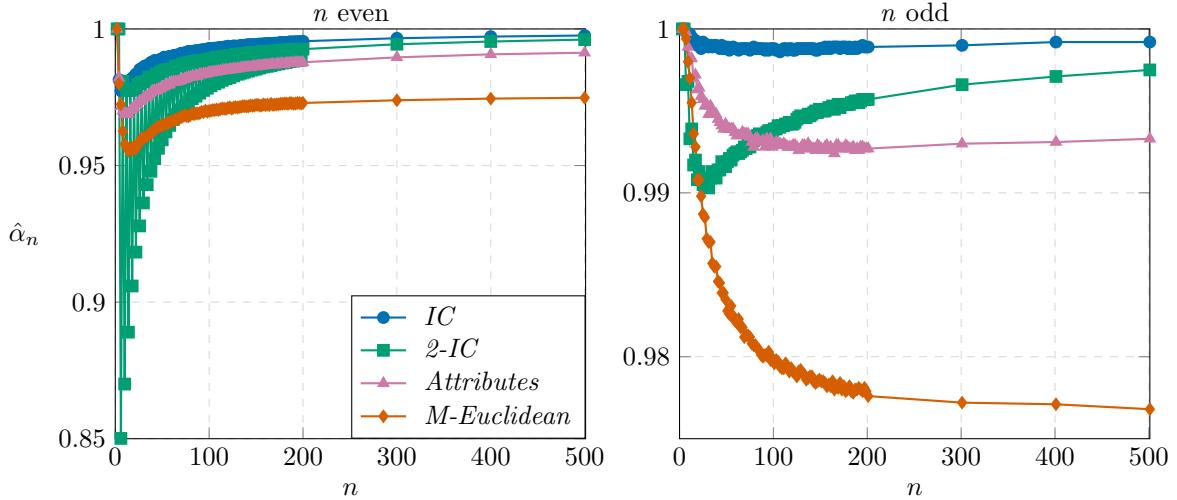
\begin{figure}[!htb]
    \centering

    % Subplot 1
    \begin{subfigure}{0.49\textwidth}
        \begin{tikzpicture}
            \begin{axis}[
                width=\textwidth,
                height=7cm,
                ymin=0.85,       % Set the minimum y-axis value
                ymax=1,         % Set the maximum y-axis value
                xmin=0,
                xmax=501,
                xlabel={$n$},
                grid=both,
                grid style={dashed, gray!30},
                cycle list name=color list,
                every axis plot/.append style={thick},
                ylabel={$\hat\alpha_n$},
                ylabel style={rotate=-90},
                title={\textit{n} even},
                title style={
                    yshift=-1.5ex  % Adjust to bring title closer to the plot
                },
                legend cell align={left},
                legend style={
                    at={(1,0)}, % Adjusts the legend position
                    anchor=south east,
                    column sep=1ex, % Space between legend columns
                },
            ]
                \addplot[acmDarkBlue, mark=*] table [x=n, y=ratio] {data/maximum/icEven.txt};
                \addplot[acmGreen, mark=square*] table [x=n, y=ratio] {data/maximum/2icEven.txt};
                \addplot[acmPink, mark=triangle*] table [x=n, y=ratio] {data/maximum/attributesEven.txt};
                \addplot[acmOrange, mark=diamond*] table [x=n, y=ratio] {data/maximum/MEuclEven.txt};
            \addlegendentry{\emph{IC}}
            \addlegendentry{\emph{2-IC}}
            \addlegendentry{\emph{Attributes}}
            \addlegendentry{\emph{M-Euclidean}}
            \end{axis}
        \end{tikzpicture}
    \end{subfigure}
    \hfill
    % Subplot 2
    \begin{subfigure}{0.49\textwidth}
        \begin{tikzpicture}
            \begin{axis}[
                width=\textwidth,
                height=7cm,
                ymin=0.975,       % Set the minimum y-axis value
                ymax=1,         % Set the maximum y-axis value
                xmin=0,
                xmax=501,
                ytick={.98, .99, 1},
                xlabel={$n$},
                grid=both,
                grid style={dashed, gray!30},
                cycle list name=color list,
                every axis plot/.append style={thick},
                title={\textit{n} odd},
                title style={
                    yshift=-1.5ex  % Adjust to bring title closer to the plot
                },
            ]
                \addplot[acmDarkBlue, mark=*] table [x=n, y=ratio] {data/maximum/icOdd.txt};
                \addplot[acmGreen, mark=square*] table [x=n, y=ratio] {data/maximum/2icOdd.txt};
                \addplot[acmPink, mark=triangle*] table [x=n, y=ratio] {data/maximum/attributesOdd.txt};
                \addplot[acmOrange, mark=diamond*] table [x=n, y=ratio] {data/maximum/MEuclOdd.txt};
            \end{axis}
        \end{tikzpicture}
    \end{subfigure}

    \caption{Ratio between the average size of a maximum stable matching and a maximum matching}
    \label{fig:maximumstablevperfect_subplots}
\end{figure}

To confirm our intuitions, we performed another experiment generating very large instances to generate estimates for $P_n$ and $\alpha_n$, the results of which can be seen in Table \ref{table:alphabig}\footnote{Due to the required computation time, these results are averaged over 3,000 rather than 7,000 instances for each (size, statistical culture) pair, as done for the experiments in Sections 3 and 4.}. While our estimate for the solvability probability is 0 for six of the eight cases considered in the table, our estimate for $\alpha_n$ is close to 1 for most cases. Interestingly, although \emph{Mallows-Euclidean} instances with $n=5001$ do not show a significant decrease in the estimate for $\alpha_n$ compared to $n=501$, it does not have an upward trend either. In fact, the values in the table suggest that for each statistical culture, their $\hat{\alpha}_n$ values for $n$ even and odd align more closely as $n$ increases.

\begin{table}[!htb]
    \centering
    \begin{tabular}{c c c c c}
        \toprule
        & \multicolumn{2}{c}{$n=5,000$} & \multicolumn{2}{c}{$n=5,001$} \\
        \cmidrule(lr){2-3} \cmidrule(lr){4-5}
        Culture & $\hat{P}_n$ & $\hat{\alpha}_n$ & $\hat{P}_n$ & $\hat{\alpha}_n$ \\
        \midrule
        \textit{IC} & 0.2663 & 0.9996 & 0.0000 & 0.9998 \\
        \textit{2-IC} & 0.1003 & 0.9994 & 0.0000 & 0.9995 \\
        \textit{Attributes} & 0.0000 & 0.9956 & 0.0000 & 0.9958 \\
        \textit{M-Euclidean} & 0.0000 & 0.9762 & 0.0000 & 0.9764 \\
        \bottomrule
    \end{tabular}
    \caption{Estimates of $P_n$ and $\alpha_n$ for very large $n$}
    \label{table:alphabig}
\end{table}

\paragraph{Finding Optimal Stable Matchings and Partitions}

Our experimental results in Section \ref{sec:numpartitions} suggest that, for a small and medium number of agents, the enumeration of reduced stable partitions is often feasible in practice. As previously mentioned, there is a large body of work on NP-hard problems that involve finding some sort of ``fair'' or ``optimal'' stable matching or stable partition in an {\sc sr} instance (see, for example, \cite{federegal, federegalapprox,gusfieldegalapprox,simola2021profilebased,CooperPhD,glitzner2024structuralalgorithmicresultsstable}). However, when given the set of reduced stable partitions for the problem instance, many types of optimal stable matching or optimal stable partition can be found through a reduced stable partition satisfying the desired criteria (without inspecting the potentially much larger set of all stable partitions, see Tables \ref{table:numstructureseven}-\ref{table:numstructuresodd}), as \citet{glitzner2024structuralalgorithmicresultsstable} showed for a selection of such problems. Now the caveat is, of course, that the set of reduced stable partitions could be exponentially large in the worst case.

In practice, however, we showed that the number of reduced stable partitions grows relatively slowly as the instance size increases, for the statistical cultures that we analysed. Referring back to Figure \ref{fig:rpgrowth_subplots}, the growth of $\vert \mathcal{RP}\vert$ actually slows on average for some statistical cultures as $n$ increases, showing a behaviour far from exponential growth. While other statistical cultures do suggest an exponential growth, note that the average number of reduced stable partitions for $n=500$ is still below 14 in all cases, rendering exhaustive enumeration very feasible in practice.

%Note that we purposefully do not provide detailed timing data for our experiments because we believe that designing fast practical implementations for the enumeration of stable matchings and partitions should be subject to further study independently from the empirical structural results presented here.

\section{Conclusion and Outlook}
\label{sec:conclusion}

This paper presents a detailed empirical analysis, complemented by new structural results, of solvable and unsolvable instances of the {\sc Stable Roommates} problem, with a focus on their internal structures and their relationship to solution concepts such as maximum stable matchings. Our key findings are summarised as follows:

\begin{itemize}
    \item The ratio of solvable versus unsolvable instances varies significantly between different statistical cultures. Notably, we observed distinct behaviors between instances with even and odd values of $n$, as well as other nuanced separations.
    \item While the total number of stable partitions can grow rapidly with $n$ (for some statistical cultures), the number of reduced stable partitions grows, on average, very slowly. This makes their enumeration feasible even for instances with up to 500 agents. For some families of instances, we even showed that stable partitions (and thus stable matchings, if any) are unique.
    \item Despite the fact that the total number of stable partitions can grow somewhat quickly with $n$, there is very little variability between the structures contained within them, measured by the number of distinct stable cycles and pairs.
    \item Despite a much higher prevalence of unsolvable instances for large $n$, the number of odd cycles (the primary cause of unsolvability) remains low, with a clear hierarchy of likelihoods favoring cycles of length 3. The growth also decays, and we gave some new estimates for the average growth function. 
    %\item Although there does not seem to be any theoretical connection between the number of odd cycles an unsolvable instance admits and the minimum number of blocking pairs any matching of the agents admits, we can observe a proportionality empirically. In fact, our results suggest that previous preliminary empirical results on almost stable matchings that established that unsolvable random instances with uniform random preferences only admit very few blocking pairs in their almost stable matchings \cite{pepe,sofiamsci} do not give a full perspective, and that sampling from other cultures can lead to much higher expected numbers.
\end{itemize}

Leveraging these insights, we demonstrated that maximum stable matchings are nearly complete in most cases, offering a practically viable solution concept despite their worst-case theoretical limitations. Our results also suggest that many NP-hard problems concerning ``optimal'' and ``fair'' stable matchings can be efficiently solved in practice due to the limited number of reduced stable partitions (and thus stable matchings).

The main open problems are clear: develop new tools to settle the key question regarding the behaviour of $P_n$ in the limit for $n$ odd and for other statistical distributions. We also propose a new problem: we introduce the parameter $\alpha_n$ capturing the expected size ratio between a maximum stable matching and an arbitrary maximum matching in the limit. Is it true that $\lim_{n\rightarrow \infty}\alpha_n < 1$, or $\lim_{n\rightarrow \infty}\alpha_n = 1$? And does this question depend on the limiting behaviour of $P_n$? To this end, we conjecture the following.

\begin{conj}
    For \emph{IC}, \emph{2-IC} and \emph{Attributes} instances, $\lim_{n\rightarrow \infty}\alpha_n = 1$.
\end{conj}

It would also be of practical interest to design fast implementations for the enumeration of reduced stable partitions, for example using low-level programming, parallelisation, or other computational tools, and investigate the speed of these techniques when solving NP-hard stable matching and stable partition problems in the {\sc sr} setting. 

The connection between stable partitions and almost stable matchings \cite{abraham06} presents another promising direction, where structural insights could be used to design improved approximation algorithms.

\paragraph{Acknowledgements} 

The authors would like to thank William Pettersson for his helpful advice and support with the cluster computations. Frederik Glitzner is supported by a Minerva Scholarship from the School of Computing Science, University of Glasgow. 
%For the purpose of open access, the authors have applied a Creative Commons Attribution (CC BY) licence to any Author Accepted Manuscript version arising from this submission.

\paragraph{Software and Data Availability}
Our code is openly available; see reference \cite{experimentsCode}. For instructions on how to generate the same instances that we used in our experiments, and to replicate our results, we refer to the {\tt README.md} file.

\bibliographystyle{reference_format}
\bibliography{papers}
%\printbibliography

\end{document}